\documentclass[11pt,a4paper]{article}
\pdfoutput=1
 
\usepackage[utf8]{inputenc}
\usepackage[bookmarks]{hyperref}
\hypersetup{colorlinks=true,citecolor=blue,linkcolor=blue,filecolor=blue,urlcolor=blue}
\usepackage{amsmath,amssymb,amsthm}
\usepackage{dsfont}
\usepackage{fullpage}

\usepackage[affil-it]{authblk}

\usepackage{cleveref}
\crefformat{enumi}{(#2#1#3)}
\crefmultiformat{enumi}{(#2#1#3)}{ and~(#2#1#3)}{, (#2#1#3)}{ and~(#2#1#3)}
\crefrangeformat{enumi}{(#3#1#4)--(#5#2#6)}

\usepackage{autonum}

\usepackage{enumerate}
\usepackage{mathtools}

\usepackage{booktabs}
 
\allowdisplaybreaks[4]

\numberwithin{equation}{section}

\newtheorem{theorem}{Theorem}[section]
\newtheorem{lemma}[theorem]{Lemma}
\newtheorem{proposition}[theorem]{Proposition}
\newtheorem{corollary}[theorem]{Corollary}
\theoremstyle{definition}
\newtheorem{definition}[theorem]{Definition}

\usepackage{tikz}
\usepackage{pgfplots}
\pgfplotsset{compat=1.10}
\usepackage{xcolor}

\usetikzlibrary{calc}

\definecolor{dgreen}{HTML}{006600}
\definecolor{lgreen}{HTML}{B3FFB3}

\usepackage{caption}
\usepackage{subcaption}

\newcommand{\cA}{\mathcal{A}}
\newcommand{\cB}{\mathcal{B}}
\newcommand{\cC}{\mathcal{C}}
\newcommand{\cD}{\mathcal{D}}

\newcommand{\cH}{\mathcal{H}}

\newcommand{\cJ}{\mathcal{J}}
\newcommand{\cK}{\mathcal{K}}

\newcommand{\cN}{\mathcal{N}}

\newcommand{\cR}{\mathcal{R}}

\newcommand{\cT}{\mathcal{T}}
\newcommand{\cU}{\mathcal{U}}

\newcommand{\trho}{\tilde{\rho}}

\newcommand{\one}{\mathds{1}}
\newcommand{\eps}{\varepsilon}

\DeclareMathOperator{\tr}{Tr}
\DeclareMathOperator{\rk}{rk}
\DeclareMathOperator{\id}{id}

\DeclareMathOperator{\supp}{supp}

\DeclareMathOperator{\op}{op}

\newcommand{\Do}{D^{(1)}_{\rightarrow}}
\newcommand{\Dt}{D^{(1)}_{\leftrightarrow}}
\newcommand{\Done}{D_{\rightarrow}}
\newcommand{\Dtwo}{D_{\leftrightarrow}}
\newcommand{\Qone}{Q^{(1)}}

\newcommand{\sumi}{\sum\nolimits}

\newcommand{\onehalf}{\frac{1}{2}}
\newcommand{\ox}{\otimes}
\DeclareMathOperator{\dg}{dg}
\newcommand{\bbF}{\mathbb{F}}
\newcommand{\tA}{\tilde{A}}
\DeclareMathOperator{\adeg}{adg}
\newcommand{\qqquad}{\qquad\qquad}

\newcommand{\Emp}{E_{\text{\normalfont MP}}}
\newcommand{\Ewd}{E_{\text{\normalfont WD}}}
\newcommand{\Em}{E_{\text{\normalfont M}}}
\newcommand{\Eda}{E_{\text{\normalfont DA}}}
\newcommand{\sep}{\text{\normalfont SEP}}
\newcommand{\ppt}{\text{\normalfont PPT}}
\newcommand{\dgb}{\text{\normalfont DEG}}
\newcommand{\adg}{\text{\normalfont ADG}}
\newcommand{\mcs}{\text{\normalfont MC}}
\newcommand{\F}{\mathbb{F}}

\newcommand{\barR}{\overline{\mathbb{R}}}
\newcommand{\vphi}{\varphi}

\usepackage[backend=bibtex8,sorting=nty,firstinits=true,doi=true,isbn=false,url=false,maxbibnames=6]{biblatex}
\renewbibmacro{in:}{}
\title{Useful states and entanglement distillation}
\author[a,b]{Felix Leditzky\thanks{Email: \texttt{felix.leditzky@jila.colorado.edu}}}
\author[c]{Nilanjana Datta\thanks{Email: \texttt{n.datta@statslab.cam.ac.uk}}}
\author[a,b,d]{Graeme Smith\thanks{Email: \texttt{gsbsmith@gmail.com}}}
\affil[a]{\small JILA, University of Colorado/NIST, Boulder, CO 80309, USA}
\affil[b]{\small Center for Theory of Quantum Matter, University of Colorado, Boulder, CO 80309, USA}
\affil[c]{\small Statistical Laboratory, Centre for Mathematical Sciences, University of Cambridge, Cambridge~CB3~0WB, UK}
\affil[d]{\small Department of Physics, University of Colorado, Boulder, CO 80309, USA}

\begin{document}
\maketitle
 
\begin{abstract}
We derive general upper bounds on the distillable entanglement of a mixed state under one-way and two-way LOCC. 
In both cases, the upper bound is based on a convex decomposition of the state into `useful' and `useless' quantum states. 
By `useful', we mean a state whose distillable entanglement is non-negative and equal to its coherent information (and thus given by a single-letter, tractable formula). 
On the other hand, `useless' states are undistillable, i.e., their distillable entanglement is zero. 
We prove that in both settings the distillable entanglement is convex on such decompositions. 
Hence, an upper bound on the distillable entanglement is obtained from the contributions of the useful states alone, being equal to the convex combination of their coherent informations. 
Optimizing over all such decompositions of the input state yields our upper bound.
The useful and useless states are given by degradable and antidegradable states in the one-way LOCC setting, and by maximally correlated and PPT states in the two-way LOCC setting, respectively.
We also illustrate how our method can be extended to quantum channels.

Interpreting our upper bound as a convex roof extension, we show that it reduces to a particularly simple, non-convex optimization problem for the classes of isotropic states and Werner states.
In the one-way LOCC setting, this non-convex optimization yields an upper bound on the quantum capacity of the qubit depolarizing channel that is strictly tighter than previously known bounds for large values of the depolarizing parameter. 
In the two-way LOCC setting, the non-convex optimization achieves the PPT-relative entropy of entanglement for both isotropic and Werner states.
\end{abstract}
 
\section{Introduction}

\subsection{Entanglement distillation}
Entanglement is an integral part of quantum information theory and quantum mechanics, acting as an indispensable resource for quantum information protocols such as teleportation \cite{BBC+93}, superdense coding \cite{BW92}, or entanglement-assisted classical \cite{BSST99} and quantum \cite{DHW04} communication through quantum channels.
In these protocols, the entanglement resource is usually assumed to have the special form of independent and identically distributed (i.i.d.) copies of an ebit $|\Phi_+\rangle \coloneqq \frac{1}{\sqrt{2}}(|00\rangle + |11\rangle)$, that is, a pure maximally entangled state between two qubits.
This assumption simplifies the aforementioned protocols and makes them amenable to a detailed theoretical analysis as well as experimental realization in the laboratory.
It is therefore important to find \emph{entanglement distillation} protocols, which convert $n$ copies of a noisy or mixed bipartite entangled state into $m_n$ ebits $\Phi_+$ to arbitrary precision with increasing $n$.

In a general entanglement distillation protocol, two parties (say, Alice and Bob) are allowed to use local operations and classical communication (LOCC).
One usually distinguishes between the following two settings: either the classical communication is restricted to only one-way communication from Alice to Bob, or two-way communication between Alice and Bob is possible. 
In both settings, Alice and Bob initially share $n$ copies of a mixed bipartite state $\rho_{AB}$, and their goal is to obtain, via one-way or two-way LOCC, a state that is close to $\Phi_+^{\ox m_n}$ with respect to a suitable distance measure (such as the purified distance \cite{TCR10}).
If the distance between the final and the target state vanishes asymptotically, then the asymptotic rate at which ebits are generated, $\lim_{n\to\infty} m_n/n$, is called an achievable rate for one-way (two-way) entanglement distillation.
The \emph{one-way distillable entanglement} $\Done(\rho_{AB})$ is defined as the supremum over all achievable rates under one-way LOCC.
Likewise, the \emph{two-way distillable entanglement} $\Dtwo(\rho_{AB})$ is defined as the supremum over all achievable rates under two-way LOCC.
Since every one-way LOCC operation is also a two-way LOCC operation, we have for all bipartite states $\rho_{AB}$ that
\begin{align}
\Done(\rho_{AB}) \leq \Dtwo(\rho_{AB}). \label{eq:one-way-two-way}
\end{align}

\textcite{DW05} proved the \emph{hashing bound}, establishing the coherent information as an achievable rate for one-way entanglement distillation (and thus also for two-way entanglement distillation): 
\begin{align}
\Done(\rho_{AB}) \geq I(A\rangle B)_\rho,\label{eq:hashing-bound}
\end{align} 
where the coherent information is defined as $I(A\rangle B)_\rho \coloneqq S(B)_\rho - S(AB)_\rho$, with the von Neumann entropy $S(A)_\rho \coloneqq -\tr(\rho_A\log\rho_A)$.
Furthermore, they derived the following regularized formulae for the distillable entanglement under one-way and two-way LOCC \cite{DW05}:
\begin{align}
\Done(\rho_{AB}) &= \lim_{n\to\infty}\frac{1}{n} \Do(\rho_{AB}^{\ox n}) \label{eq:one-way-dist-entanglement}\\ 
\Dtwo(\rho_{AB}) &= \lim_{n\to\infty}\frac{1}{n} \Dt(\rho_{AB}^{\ox n}).\label{eq:two-way-dist-entanglement}
\end{align}
Here, $D^{(1)}_*(\cdot)$ for $*\in\lbrace \rightarrow, \leftrightarrow\rbrace$ is defined as
\begin{align}
D^{(1)}_*(\rho_{AB}) \coloneqq \max_{\Lambda\colon AB \to A'B'} I(A'\rangle B')_{\Lambda(\rho)},
\end{align}
where the maximization is over one-way and two-way LOCC operations $\Lambda\colon AB\to A'B'$, respectively.

Similar to the quantum capacity, the regularizations in \eqref{eq:one-way-dist-entanglement} and \eqref{eq:two-way-dist-entanglement} render the distillable entanglement intractable to compute in most cases.
Hence, it is desirable to identify classes of bipartite states for which the formulae in \eqref{eq:one-way-dist-entanglement} and \eqref{eq:two-way-dist-entanglement} reduce to single-letter formulae that can be easily computed.
Moreover, we are interested in computable upper bounds on $\Done(\rho_{AB})$ and $\Dtwo(\rho_{AB})$ for arbitrary bipartite states.
We address both problems in the present paper.

\subsection{Method and main results}\label{sec:main-results}

To obtain computable upper bounds on the regularized formulae \eqref{eq:one-way-dist-entanglement} and \eqref{eq:two-way-dist-entanglement} for the distillable entanglement under one-way and two-way LOCC, we first identify classes of `useful' and `useless' states in both settings.
Here, we call a state $\rho_{AB}$ useful, if $D^{(1)}_*(\rho_{AB})$ is equal to the coherent information $I(A\rangle B)_\rho$ for $*\in\lbrace \rightarrow,\leftrightarrow\rbrace$, and thus additive on tensor products $\rho^{\ox n}_{AB}$.
It then follows immediately from \eqref{eq:one-way-dist-entanglement} and \eqref{eq:two-way-dist-entanglement} that also $D_*(\rho_{AB}) = I(A\rangle B)_\rho$.
In the one-way setting, the useful states are \emph{degradable states} ($\dgb$) (cf.~\Cref{def:degradable-states}), while in the two-way setting the useful states are
\emph{maximally correlated states} $(\mcs)$ (cf.~\Cref{def:mcs}).
Note that we have $\mcs\subseteq \dgb$.

On the other hand, useless states $\sigma_{AB}$ are such that $D^{(1)}_*(\sigma_{AB}^{\ox n})$ is zero for all $n\in\mathbb{N}$, from which $D_*(\sigma_{AB})=0$ follows. 
The class of useless states is given by \emph{antidegradable states} (cf.~\Cref{def:degradable-states}) in the one-way setting, and by states with positive partial transpose (or $\ppt$ states for short) in the two-way setting.
We list the four classes of states in \Cref{tab:useful-states} below.

\begin{table}[ht]

	\centering
	\begin{tabular}{ccc}
		\toprule
		& useful & useless\\
		\midrule \addlinespace[0.5em]
		1-way & $\dgb$ & $\adg$\\
		\addlinespace[0.5em]
		2-way & $\mcs$ & $\ppt$\\
		\bottomrule
	\end{tabular}\hspace{3em}
	\begin{tabular}{rl}
		$\dgb$ & degradable\\
		$\adg$ & antidegradable\\
		$\mcs$ & maximally correlated\\
		$\ppt$ & positive partial transpose
	\end{tabular}
	\caption{Useful and useless states for one-way and two-way entanglement distillation.}
	\label{tab:useful-states}
\end{table}

The crucial step in proving our main results is to observe that $D_*(\cdot)$ is convex on convex combinations of the corresponding useful and useless states.
This is proved in \Cref{prop:convexity-one-way} for the one-way setting by adapting an argument by \textcite{WP07}, and in \Cref{prop:convexity-two-way} for the two-way setting inspired by an argument by \textcite{Rai99}. 
Together with the known values of the distillable entanglement on useful and useless states (given by their coherent information and 0, respectively), this proves the upper bounds on the one-way distillable entanglement in \Cref{thm:one-way-dist-upper-bound} and on the two-way distillable entanglement in \Cref{thm:two-way-dist-upper-bound}, which constitute our main result.
We note that in both settings the class of useful states includes all pure quantum states (that is, every pure state is both degradable and maximally correlated).
Hence, any pure-state ensemble of a bipartite state yields a decomposition into useful states.
In both settings, the optimal such pure-state ensemble yields the entanglement of formation, and our upper bounds can be understood as an improvement over the latter.
Moreover, in the one-way setting our result can be straightforwardly extended to quantum channels, yielding an analogous upper bound on the quantum capacity of a quantum channel in \Cref{thm:quantum-capacity-upper-bound} that was first reported by \textcite{Yang}.

Finally, we focus on the distillable entanglement of isotropic states and Werner states.
Interpreting our upper bounds on the distillable entanglement as convex roof extensions allows us to use a result by \textcite{VW01} that exploits the symmetries of isotropic states and Werner states to facilitate the computation of the convex roof extension.
The result is a simplification of our upper bound to a (non-convex) optimization problem that can be solved numerically for small dimensions.
In particular, this yields an upper bound on the quantum capacity of the qubit depolarizing channel that is tighter than the best previously known upper bound for large values of the depolarizing parameter.

The rest of this paper is structured as follows.
We first fix some notation in \Cref{sec:notation}.
We then dedicate \Cref{sec:one-way} to developing the method outlined above for one-way entanglement distillation.
Furthermore, we introduce and discuss the notion of approximately (anti)degradable states in \Cref{sec:approx-degradability}, which is inspired by and analogous to the notion of approximately degradable quantum channels in \cite{SSWR15}.
The derivation of our main result for two-way entanglement distillation is carried out in \Cref{sec:two-way}.
Apart from the results mentioned above, we also discuss a method for constructing decompositions into maximally correlated states via the generalized Bell basis in \Cref{sec:blocking}.
In \Cref{sec:symmetries} we derive the non-convex optimization form of our upper bounds for isotropic and Werner states.
Finally, we give some concluding remarks in \Cref{sec:conclusion}.
\Cref{sec:antidegradable-states} contains a discussion of antidegradable states and their maximal overlap with maximally entangled states.

\subsection{Notation}\label{sec:notation}
Throughout the paper we only consider finite-dimensional Hilbert spaces.
For Hilbert spaces $\cH_1$ and $\cH_2$, we denote by $\cB(\cH_1,\cH_2)$ the set of linear maps from $\cH_1$ to $\cH_2$, and we write $\cB(\cH) = \cB(\cH,\cH)$ for the algebra of linear operators on a single Hilbert space $\cH$.
Upper-case indices are used to label quantum systems: for a Hilbert space $\cH_A$ corresponding to a quantum system $A$, we write $|\psi\rangle_A\in\cH_A$ and $\rho_A\in\cB(\cH_A)$, and we use the notation $\cH_{A_1A_2\dots} \coloneqq \cH_{A_1}\ox\cH_{A_2}\ox\dots$.
We write $|A|\coloneqq \dim\cH_A$ for the dimension of a quantum system $A$ with associated Hilbert space $\cH_A$, and $\rk\rho_A$ for the rank of the operator $\rho_A$.
We use the shorthand $A\cong B$ to indicate that the Hilbert spaces associated to $A$ and $B$ are isomorphic, $\cH_A \cong \cH_B$.
A \emph{quantum state} (or simply state) is an operator $\rho_A\in\cB(\cH_A)$ with $\rho_A\geq 0$ and $\tr\rho_A = 1$, and we denote the set of states on $\cH_A$ by $\cD(\cH_A)$.
We write $\psi_A\equiv |\psi\rangle\langle \psi|_A\in\cB(\cH_A)$ for the rank-1 projector associated to the pure state $|\psi\rangle_A\in\cH_A$.

The \emph{von Neumann entropy} of a state $\rho_A$ is defined by $S(A)_\rho\coloneqq -\tr(\rho_A\log\rho_A)$, the \emph{coherent information} of a bipartite state $\rho_{AB}$ by $I(A\rangle B)_\rho\coloneqq S(B)_\rho - S(AB)_\rho$, and the \emph{conditional entropy} by $S(A|B)_\rho \coloneqq - I(A\rangle B)_\rho$.
For a probability distribution $\lbrace p_i\rbrace_i$, the \emph{Shannon entropy} is defined by $H(\lbrace p_i\rbrace_i)\coloneqq - \sum_i p_i \log p_i$.
For $p\in[0,1]$, the \emph{binary entropy} is defined by $h(p)\coloneqq -p\log p - (1-p)\log (1-p)$.
All exponentials and logarithms are taken to base $2$.

A quantum channel $\cN\colon \cB(\cH)\to\cB(\cK)$ is a linear, completely positive (CP), trace-preserving (TP) map between the algebras $\cB(\cH)$ and $\cB(\cK)$ of linear operators on Hilbert spaces $\cH$ and $\cK$.
We write $\cN\colon A\to B$ for a quantum channel from $\cB(\cH_A)$ to $\cB(\cH_B)$.
Let $\cN(\rho_A)=\tr_E(V\rho_A V^\dagger)$ be the Stinespring representation of $\cN$ with the isometry $V\colon \cH_{A}\to \cH_{B}\ox\cH_{E}$.
Then the complementary channel $\cN^c\colon A\to E$ is defined by $\cN^c(\rho_A)\coloneqq \tr_B(V\rho_A V^\dagger)$.
We often omit the identity map denoted by $\id$, i.e., for a map $T\colon A\to A'$ acting on the $A$ part of a state $\rho_{AB}$, we also write $T(\rho_{AB})$ instead of $(T\ox\id_B)(\rho_{AB})$.

Let $\lbrace |\Phi_{n,m}\rangle \rbrace_{n,m=0,\dots,d-1}$ be the generalized Bell basis defined as follows.
We define the generalized Pauli operators $X$ and $Z$ via their action on a fixed basis $\lbrace |k\rangle \rbrace_{k=0}^{d-1}$ of $\mathbb{C}^d$,
\begin{align}
X|k\rangle &\coloneqq |k+1 (\text{mod }d)\rangle & Z |k\rangle \coloneqq \omega^k |k\rangle, \label{eq:gen-pauli-operators}
\end{align}
where $\omega \coloneqq \exp(2\pi i/d)$ is a $d$-th root of unity.
The generalized Pauli operators satisfy $XZ = \omega ZX$.
Setting $|\Phi_+\rangle \coloneqq \frac{1}{\sqrt{d}}\sum_{i=0}^{d-1} |ii\rangle$, we define 
\begin{align}
|\Phi_{n,m}\rangle \coloneqq (\one_d \ox X^m Z^n) |\Phi_+\rangle,\label{eq:bell-states}
\end{align}
which satisfy $\langle \Phi_{n,m}|\Phi_{n',m'}\rangle = \delta_{n,n'}\delta_{m,m'}$.

Finally, for a vector $|\psi\rangle = \sum_{i,j} \lambda_{ij} |i\rangle_A \ox |j\rangle_B \in \cH_A\ox\cH_B$, we define an associated operator $\op(\psi_{AB})\in\cB(\cH_B,\cH_A)$ by
\begin{align}
\op(\psi_{AB}) \coloneqq \sum_{i,j} \lambda_{ij} |i\rangle_A \langle j|_B.\label{eq:vec-op}
\end{align}

\section{One-way entanglement distillation}\label{sec:one-way}

\subsection{Operational setting}
Given a mixed bipartite state $\rho_{AB}$, the \emph{one-way distillable entanglement} $\Done(\rho_{AB})$ is defined as the optimal rate of distilling ebits from many copies of $\rho_{AB}$ via local operations and forward (or one-way) classical communication (LOCC) from Alice to Bob.

A general one-way LOCC operation can be modeled as a \emph{quantum instrument} $T\colon A\to A'M$, defined by
\begin{align}
T(\theta_{A}) \coloneqq \sumi_m T_m(\theta_{A}) \ox |m\rangle\langle m|_M,
\end{align}
where $\lbrace |m\rangle\rbrace_m$ is an orthonormal basis for the classical register $M$, and for each $m$ the map $T_m\colon A\to A'$ is CP such that $\sum_m T_m$ is TP.

As mentioned in the introduction, \textcite{DW05} derived the following regularized formula for the one-way distillable entanglement:
\begin{align}
\Done(\rho_{AB}) &= \lim_{n\to\infty}\frac{1}{n} \Do(\rho_{AB}^{\ox n}),\label{eq:one-way-dist-entanglement2}
\end{align}
where $\Do(\rho_{AB})$ can be expressed as
\begin{align}
\Do(\rho_{AB})&\coloneqq \max_T \sumi_m \lambda_m I(A'\rangle B)_{\rho_m}.\label{eq:D-one}
\end{align} 
Here, the maximization is over instruments $T\colon A\to A'M$, and we set $\rho_m\coloneqq \frac{1}{\lambda_m}T_m(\rho_{AB})$, with $\lambda_m\coloneqq \tr(T_m(\rho_{AB}))$ denoting the probability of obtaining the outcome $m$ of $T$.
Equivalently, \eqref{eq:D-one} can be written as
\begin{align}
\Do(\rho_{AB}) = \max_T I(A'\rangle BM)_{T(\rho_{AB})}.\label{eq:D-one-compact}
\end{align}

Instead of maximizing the coherent information in \eqref{eq:D-one-compact} over instruments $T$, it can be more convenient to consider a maximization over isometric extensions of an instrument in the following way.
First, we note that it suffices to consider instruments $T = \sum_m T_m \ox |m\rangle \langle m|$ where each of the CP maps $T_m$ has only one Kraus operator, i.e., $T_m(\cdot) = K_m \cdot K_m^\dagger$ for each $m$ and operators $K_m\colon A\to A'$ \cite{DW05}.
In this case, an isometric extension $V\colon A\to A'MN$ can be defined as\footnote{In the general case where each TP map $T_m$ might have more than one Kraus operator, say $\lbrace K_{m,j}\rbrace_j$, an isometric extension can be defined by including an additional system $F$ with orthonormal basis $\lbrace |\phi_{m,j}\rangle_F\rbrace_{m,j}$ that acts as the environment for each $T_m$.}
\begin{align}
V \coloneqq \sum_m K_m \ox |m\rangle_M \ox |m\rangle_N\label{eq:instrument-isometry}
\end{align}
for a classical register $N\cong M$.
Since $\sum_m T_m = \sum_m K_m\cdot K_m^\dagger$ is TP by definition of a quantum instrument, we have $\sum_m K_m^\dagger K_m = \one_A$, which implies $V^\dagger V = \one_A$.
Hence, $V$ is indeed an isometry, and we have $T(\rho_A) = \tr_N(V\rho_A V^\dagger)$ for all $\rho_A$.
Using \eqref{eq:instrument-isometry}, we can write \eqref{eq:D-one-compact} as
\begin{align}\label{eq:D-one-instrument}
\Do(\rho_{AB}) = \max_V I(A'\rangle BM)_{\omega},
\end{align}
where $\omega_{A'BM} = \tr_N(V \rho_{AB} V^\dagger )$.

\begin{lemma}
\label{lem:positivity-of-D-one}
$\Do(\rho_{AB})\geq 0$ for all bipartite states $\rho_{AB}$.
\end{lemma}

\begin{proof}
Since $\Do(\rho_{AB})$ can be expressed as a maximization over all instrument isometries $V$ of the form \eqref{eq:instrument-isometry} as stated in \eqref{eq:D-one-instrument}, the lemma is proved by constructing a particular $V$ for which we obtain $I(A'\rangle BM)_{\omega} =0$ with $\omega_{A'BM} = \tr_N(V\rho_{AB}V^\dagger)$.

To this end, let $|\phi\rangle_{ABE}$ be a purification of $\rho_{AB}$, and consider a Schmidt decomposition of $|\phi\rangle_{ABE}$ with respect to the bipartition $A|BE$,
\begin{align}
|\phi\rangle_{ABE} \coloneqq \sumi_i \lambda_i |i\rangle_A |i\rangle_{BE},
\end{align}
where the Schmidt coefficients $\lambda_i\geq 0$ for all $i$.
We define the instrument isometry $V\colon A\to A'MN$, 
\begin{align}
V \coloneqq \sumi_i |i\rangle_{A'}\langle i|_A \ox |i\rangle_M \ox |i\rangle_N,
\end{align}
where $A'\cong A$.
Applying $V$ to the purification $|\phi\rangle_{ABE}$ of $\rho_{AB}$, we obtain the pure state
\begin{align}
|\omega\rangle_{A'MNBE} &= V|\phi\rangle_{ABE} = \sumi_i \lambda_i |iii\rangle_{A'MN}|i\rangle_{BE},
\end{align}
whose marginals $\omega_{BM}$ and $\omega_{EN}$ are given by
\begin{align}
\omega_{BM} &= \sumi_i \lambda_i^2 |i\rangle\langle i|_M \ox \tr_E|i\rangle\langle i|_{BE} & \omega_{EN} &= \sumi_i \lambda_i^2 |i\rangle\langle i|_N \ox \tr_B|i\rangle\langle i|_{BE}.
\end{align}
Evaluating the coherent information of the state $\omega_{A'BM}$ yields
\begin{align}
I(A'\rangle BM)_\omega &= S(BM)_\omega - S(A'BM)_\omega\\
&= S(BM)_\omega - S(EN)_\omega\\
&= \sumi_i \lambda_i^2 \left(S(\tr_E|i\rangle\langle i|_{BE}) - S(\tr_B|i\rangle\langle i|_{BE})\right)\\
&= 0,
\end{align}
which proves the claim.
\end{proof}

\subsection{(Conjugate) degradable and antidegradable states}\label{sec:definitions}

We now define the classes of `useful' and `useless' states for one-way entanglement distillation, as explained in \Cref{sec:main-results}.

\begin{definition}\label{def:degradable-states}
Let $\rho_{AB}$ be a bipartite state with purification $|\phi\rangle_{ABE}$.
The state $\rho_{AB}$ is called:
\begin{enumerate}[{\normalfont (i)}]
\item \emph{degradable}, if there is an isometry $U\colon B\to E' G$ with $E'\cong E$ such that for the state $|\varphi\rangle_{AE' GE} = U |\phi\rangle_{ABE}$ we have
\begin{align}\label{eq:degradability}
\varphi_{AE} = \varphi_{AE'} = \phi_{AE};
\end{align}

\item \emph{conjugate degradable}, if \eqref{eq:degradability} holds up to complex conjugation, that is,
\begin{align}
\varphi_{AE} = \cC(\varphi_{AE'}) = \phi_{AE},
\end{align}
where $\cC$ denotes entry-wise complex conjugation with respect to a fixed basis of $E'\cong E$;

\item \emph{antidegradable}, if there is an isometry $V\colon E\to B'F$ with $B'\cong B$ such that for the state $|\psi\rangle_{ABB'F} =  V|\phi\rangle_{ABE}$ we have
\begin{align}
\psi_{AB'} = \psi_{AB} = \phi_{AB}.
\end{align}
\end{enumerate} 
\end{definition}

We note that \Cref{def:degradable-states} is independent of the chosen purification of $\rho_{AB}$, since any two purifications of $\rho_{AB}$ are related by an isometry acting only on the purifying systems.
We can then compose the (conjugate) (anti)degrading isometries from \Cref{def:degradable-states} with the isometry relating the different purifications.

The coherent information of a degradable state $\rho_{AB}$ is non-negative, $I(A\rangle B)_\rho\geq 0$, since
\begin{align}
I(A\rangle B)_\rho = I(A\rangle E'G)_\vphi \geq I(A\rangle E')_\vphi = I(A\rangle E)_\phi = - I(A\rangle B)_\rho,
\end{align}
where we used the data processing inequality for the coherent information in the first inequality, and the duality relation $I(A\rangle B)_\psi = - I(A\rangle E)_\psi$ for a pure state $|\psi\rangle_{ABE}$ in the last equality.
Using a similar argument, an antidegradable state $\sigma_{AB}$ has non-positive coherent information, $I(A\rangle B)_\sigma\leq 0$.
Symmetric states, which are both degradable and antidegradable, therefore have zero coherent information.

Every pure state $|\psi\rangle_{AB}$ is degradable, which can be seen by choosing an arbitrary purification $|\phi\rangle_{ABE} = |\psi\rangle_{AB} \ox |\chi\rangle_E$ with some pure state $|\chi\rangle_E$, and considering the isometry $U$ defined by $U |\theta\rangle_B \coloneqq |\theta\rangle_B \ox |\chi\rangle_E$.
A large class of mixed (conjugate) (anti)degradable states can be obtained from (conjugate) (anti)degradable quantum channels.
We call a quantum channel $\cN\colon A\to B$ \emph{degradable}, if there exists a quantum channel $\cD\colon B\to E$ (called a degrading map) such that 
\begin{align}\label{eq:degradable-channel}
\cN^c = \cD\circ\cN.
\end{align}
The channel $\cN$ is \emph{conjugate degradable} \cite{BDHM10}, if instead of \eqref{eq:degradable-channel} we have
\begin{align}\label{eq:conjugate-degradable-channel}
\cC \circ \cN^c = \cD\circ\cN,
\end{align}
where $\cC$ denotes entry-wise complex conjugation with respect to a fixed basis as in \Cref{def:degradable-states}.
Finally, a channel $\cN$ is called \emph{antidegradable}, if there exists a quantum channel $\cA\colon E\to B$ (called an antidegrading map) such that
\begin{align}
\cN = \cA\circ \cN^c.
\end{align}
Let now $|\Phi\rangle_{A'A}$ be a maximally entangled state between $A'\cong A$ and $A$, then the \emph{Choi state} $\tau_{A'B}$ of $\cN\colon A\to B$ is defined as
\begin{align}\label{eq:choi-state}
\tau_{A'B} \coloneqq \cN(\Phi_{A'A}).
\end{align}
Similarly, we define the Choi state $\tau_{A'E}= \cN^c(\Phi_{A'A})$ of the complementary channel $\cN^c$.
The following result is obvious:
\begin{lemma}\label{lem:deg-states-channels}
Let $\cN\colon A\to B$ be a quantum channel.
Then the Choi state $\tau_{A'B}$ as defined in \eqref{eq:choi-state} is (conjugate) (anti)degradable if and only if $\cN$ is (conjugate) (anti)degradable.
\end{lemma}

Finally, we note that we occasionally simplify \Cref{def:degradable-states} to the following (equivalent) form, which is closely related to the channel picture above: a state $\rho_{AB}$ with purification $|\phi\rangle_{ABE}$ and `complementary state' $\rho_{AE} \coloneqq \tr_B\phi_{ABE}$ is degradable if there exists a CPTP degrading map $\cD\colon B\to E$ such that $\rho_{AE} = \cD(\rho_{AB})$.
In this case, the degrading isometry $U$ from \Cref{def:degradable-states} can be chosen as the Stinespring isometry (cf.~\Cref{sec:notation}) of $\cD$.
Conversely, every degrading isometry as in \Cref{def:degradable-states} gives rise to a degrading map by defining $\cD(\cdot)\coloneqq \tr_G(V\cdot V^\dagger)$ and identifying $E$ with $E'$.
We also use analogous simplifications in the case of conjugate degradability and antidegradability.

\subsection{Upper bounds on the one-way distillable entanglement}\label{sec:upper-bounds-one-way}

The hashing bound \eqref{eq:hashing-bound} states that for any state $\rho_{AB}$ the coherent information $I(A\rangle B)_\rho$ is an achievable rate for one-way entanglement distillation.
The first result of this section shows that for (conjugate) degradable states the coherent information is the \emph{optimal} rate for entanglement distillation:

\begin{proposition}\label{prop:dist-ent-single-letter}
Let $\rho_{AB}$ be a (conjugate) degradable state.
Then $\Do(\rho_{AB})$ is equal to the coherent information $I(A\rangle B)_\rho$ and thus additive: for all $n\in\mathbb{N}$,
\begin{align}
\Do(\rho_{AB}^{\ox n}) = n \Do(\rho_{AB}) = nI(A\rangle B)_\rho. 
\end{align}
Hence, the one-way distillable entanglement of $\rho_{AB}$ is equal to the coherent information,
\begin{align}\label{eq:dist-ent-single-letter}
\Done(\rho_{AB}) = I(A\rangle B)_\rho.
\end{align}
\end{proposition}
\begin{proof}
Let us first assume that $\rho_{AB}$ is degradable, that is, we have
\begin{align}\label{eq:degradable}
\varphi_{AE} = \varphi_{AE'}
\end{align}
where $|\varphi\rangle_{AE'GE} = W|\phi\rangle_{ABE}$ and $W\colon B\to E'G$ is a degrading isometry.
Let us furthermore define the following pure states:
\begin{align}
|\omega\rangle_{A'MNBE} &\coloneqq V |\phi\rangle_{ABE} \label{eq:omega}\\
|\sigma\rangle_{A'MNE'GE} &\coloneqq W |\omega\rangle_{A'MNBE} = V|\varphi\rangle_{AE'GE},\label{eq:sigma}
\end{align}
where $V\colon A \to A'MN$ is given as in \eqref{eq:instrument-isometry}. 
Consider now the following steps:
\begin{align}
I(A' \rangle BM)_\omega &= I(A' \rangle E'GM)_\sigma\\
 &=S(E'GM)_\sigma - S(A'E'GM)_\sigma\\
&= S(E'GM)_\sigma - S(NE)_\sigma \\
&= S(E'GM)_\sigma - S(ME)_\sigma \\
&= S(E'GM)_\sigma - S(ME')_\sigma \label{eq:apply-deg-1}\\
&= S(G|E'M)_\sigma \\
&\leq S(G|E')_\sigma \\
&= S(GE')_\sigma - S(E')_\sigma \\
&=  S(GE')_\sigma - S(E)_\sigma \label{eq:apply-deg-2}\\
&=  S(GE')_\sigma - S(A'MNE'G)_\sigma \\
&= S(B)_\rho - S(AB)_\rho \\
&= I(A\rangle B)_\rho
\end{align}
where the third line follows from the fact that $\sigma_{A'MNE'GE}$ is a pure state, the fourth line follows from the symmetry in $M$ and $N$ (which is evident from the definition \eqref{eq:instrument-isometry} of the isometry $V$), the fifth line follows from the degradability \eqref{eq:degradable} of the state $\rho_{AB}$, the seventh line follows from the fact that conditioning reduces entropy, the ninth line follows again from the degradability of $\rho_{AB}$, and the tenth line follows from the fact that $\sigma_{A'MNE'GE}$ is pure.

Hence, the trivial isometry achieves the maximum in $\max_V I(A'\rangle BM)_\omega$, and $\Do(\rho_{AB}) = I(A\rangle B)_\rho$.
Since the coherent information is additive on tensor products, we have $\Do(\rho_{AB}^{\ox n}) = n I(A\rangle B)_\rho$, and \eqref{eq:dist-ent-single-letter} follows from \eqref{eq:one-way-dist-entanglement}.

If $\rho_{AB}$ is only conjugate degradable, \eqref{eq:degradable} is replaced by
\begin{align}\label{eq:conjguate-degradable}
\varphi_{AE} = \cC(\varphi_{AE'}),
\end{align}
where $\cC$ denotes entry-wise complex conjugation with respect to a fixed basis.
Note that both $\sigma_{ME'}$ and $\sigma_{ME}$ are classical-quantum states, that is, they are of the form $\sum_m p_m |m\rangle\langle m|_M\ox \tau_{E}^m$ and $\sum_m p_m |m\rangle\langle m|_M\ox \tau_{E'}^m$, respectively, with $\tau_E^m = \cC(\tau_{E'}^m)$ for all $m$.
Hence, we can use \eqref{eq:conjguate-degradable} instead of \eqref{eq:degradable} in steps \eqref{eq:apply-deg-1} and \eqref{eq:apply-deg-2} above, since the von Neumann entropy is invariant under complex conjugation.
This yields the claim in the case of conjugate degradability of $\rho_{AB}$.
\end{proof}

The following lemma shows the well-known fact (see e.g.~\cite{BDSW96}) that the one-way distillable entanglement of antidegradable states is $0$.
We provide a short proof for the sake of completeness.
\begin{lemma}\label{lem:antideg-distill}
Let $\sigma_{AB}$ be an antidegradable state.
Then $\Do(\sigma_{AB}) = 0$ and $D_\rightarrow(\sigma_{AB}) = 0$.
\end{lemma}
\begin{proof}
Let $|\psi^\sigma\rangle_{ABE}$ be a purification of $\sigma_{AB}$, and denote by $\cA\colon E\to B$ the antidegrading map satisfying $\sigma_{AB} = \cA(\sigma_{AE})$, where $\sigma_{AE}=\tr_B\psi^\sigma$.
Let $V\colon A\to A'MN$ be an arbitrary isometry of the form \eqref{eq:instrument-isometry}, and consider the following steps for the coherent information evaluated on the state $V\sigma_{AB} V^\dagger$:
\begin{align}
I(A'\rangle BM) &\leq I(A'\rangle EM)\\
&= -I(A'\rangle BN)\\
&= - I(A'\rangle BM),
\end{align}
where we used data processing with respect to $\cA$ in the inequality, duality for the coherent information in the first equality, and symmetry in $M\leftrightarrow N$ in the second inequality.
It follows that $I(A'\rangle BM) \leq 0$ for any instrument isometry $V$, and hence, $\Do(\sigma_{AB})\leq 0$.
Together with \Cref{lem:positivity-of-D-one}, this proves $\Do(\sigma_{AB}) = 0$ for all antidegradable states $\sigma_{AB}$.
Since $\sigma_{AB}^{\ox n}$ is antidegradable for all $n\in\mathbb{N}$ (with the antidegrading map given by $\cA^{\ox n}$), we then have $D_\rightarrow(\sigma_{AB})=0$.
\end{proof}

We now derive a general upper bound on the one-way distillable entanglement of arbitrary, not necessarily degradable, bipartite states.
To this end, we first prove the following proposition, which shows that we can ignore the contributions from antidegradable states for the $\Do(\cdot)$ quantity.

\begin{proposition}\label{prop:forget-antidegradable}
Let $\rho_{A_1B_1}$ be degradable and $\sigma_{A_2B_2}$ be antidegradable.
Then
\begin{align}
\Do(\rho_{A_1B_1}\ox\sigma_{A_2B_2}) = \Do(\rho_{A_1B_1}).
\end{align}
\end{proposition}

\begin{proof}
We first observe that $\Do(\rho_{A_1B_1}\ox\sigma_{A_2B_2}) \geq \Do(\rho_{A_1B_1})$ holds for any two states $\rho_{A_1B_1}$ and $\sigma_{A_2B_2}$ not necessarily degradable and antidegradable. This follows from extending an optimal instrument for $\rho_{A_1B_1}$ trivially to $A_2$ and using the data processing inequality for the coherent information with respect to tracing out the $B_2$ system.

To prove the other inequality, let $\rho_{A_1B_1}$ with purification $|\psi^\rho\rangle_{A_1B_1E_1}$ be degradable with degrading map $\cD\colon B_1\to E_1$, and let $\sigma_{A_2B_2}$ with purification $|\psi^\sigma\rangle_{A_2B_2E_2}$ be antidegradable with antidegrading isometry $W\colon E_2\to B_2'E_2'$ such that $|\tau\rangle_{A_2B_2B_2'E_2'} \coloneqq W |\psi^\sigma\rangle_{A_2B_2E_2}$ satisfies
\begin{align}
\tau_{A_2B_2'} = \tau_{A_2B_2} = \sigma_{A_2B_2}.
\end{align}
Denoting by $\bbF_{B_2B_2'}$ the swap operator exchanging $B_2$ and $B_2'$, we define the state
\begin{align}
|\Omega\rangle _{A_2B_2B_2'E_2'C_B} \coloneqq \frac{1}{\sqrt{2}} \left(|\tau\rangle_{A_2B_2B_2'E_2'}\ox |0\rangle_{C_B} + \bbF_{B_2B_2'}|\tau\rangle_{A_2B_2B_2'E_2'}\ox |1\rangle_{C_B}\right)\!,
\end{align}
which satisfies $|\Omega\rangle = \bbF_{B_2B_2'}\ox X_{C_B}|\Omega\rangle$ and $\Omega_{A_2B_2} = \sigma_{A_2B_2}$.
Here, $X_{C_B}$ denotes the Pauli $X$ operator on the $C_B$ system.
Let $T\colon A_1A_2 \to A'M$ be an arbitrary instrument with isometry $V\colon A_1A_2 \to A'MN$, then we have
\begin{multline}
I(B_1; MB_2)_{T(\rho\ox \tau)} \geq I(E_1; N B_2')_{(\cD\circ T)(\rho\ox\tau)} \\ \Longleftrightarrow S(B_1) + S(MB_2) - S(E_1) - S(NB_2') \geq S(MB_1B_2) - S(NE_1B_2')\label{eq:dpi}
\end{multline}
by the data processing inequality for the mutual information with respect to $\cD$, and because
\begin{align}
S(MB_2)_{T(\rho\ox\tau)} = S(NB_2')_{T(\rho\ox\tau)}. \label{eq:MNB2}
\end{align}
Consider now the following steps:
\begin{align}
I(A' \rangle MB_1 B_2)_{T(\rho\ox\sigma)} &= S(MB_1B_2)_{T(\rho\ox\sigma)} - S(A'MB_1B_2)_{T(\rho\ox\sigma)}\\
&= S(MB_1B_2) - S(NE_1B_2'E_2'C_B)_{T(\rho\ox\tau)}\\
&= S(MB_1B_2) - S(NE_1B_2') + S(NE_1B_2') - S(NE_1B_2'E_2'C_B)\\
&\leq S(B_1) - S(E_1) + S(MB_2) - S(NB_2') + I(E_2'C_B\rangle NE_1B_2')\\
&= I(A_1\rangle B_1) + I(E_2'C_B\rangle NE_1B_2'),\label{eq:last-step}
\end{align}
where we used \eqref{eq:dpi} in the inequality, and once again \eqref{eq:MNB2} in the last equality.
For the second coherent information in \eqref{eq:last-step}, observe that
\begin{align}
I(E_2'C_B\rangle NE_1B_2') &\leq I(E_2'C_B\rangle N B_1 B_2')\\
&= I(E_2'C_B\rangle M B_1 B_2)\\
&= - I(E_2'C_B\rangle A' N E_1 B_2')\\
&\leq - I(E_2'C_B\rangle NE_1 B_2').
\end{align}
Here, the first and second inequality follow from the data processing inequality for the coherent information with respect to $\cD$ and partial trace over $A'$, respectively.
The second line follows from symmetry of $V(|\psi^\rho\rangle \ox |\Omega\rangle)$ in $M\leftrightarrow N$, and from the invariance of the coherent information under the local unitary $\mathbb{F}_{B_2B_2'}\ox X_{C_B}$.
Hence, $I(E_2'C_B\rangle NE_1B_2')\leq 0$, and \eqref{eq:last-step} yields
\begin{align}
\Do(\rho_{A_1B_1}\ox\sigma_{A_2B_2}) = \max_T I(A' \rangle MB_1 B_2) \leq I(A_1\rangle B_1) = \Do(\rho_{A_1B_1}),
\end{align}
which we set out to prove.
\end{proof}

The last ingredient for our general upper bound on the one-way distillable entanglement $\Done(\cdot)$ is \Cref{prop:convexity-one-way} below, which establishes that $\Done(\cdot)$ is convex on mixtures of states whose tensor products have subadditive $\Do(\cdot)$.
This result is analogous to the corresponding property of the quantum capacity proved by Wolf and Pérez-García \cite{WP07}, and our proof of \Cref{prop:convexity-one-way} closely follows the one given in \cite{WP07}.
We introduce the following notation: For a binary string $w^n = (w_1,\dots,w_n)\in\lbrace 0,1\rbrace^n$, we denote by $|w^n|\coloneqq |\lbrace i\colon w_i=1\rbrace|$ the \emph{Hamming weight} of $w^n$, i.e.~the number of $1$'s in $w^n$.
For states $\rho_0$ and $\rho_1$ and $w^n = (w_1,\dots,w_n)\in\lbrace 0,1\rbrace^n$, we set $\rho_{w^n}\coloneqq\rho_{w_1}\ox\dots\ox\rho_{w_n}$.
We then have the following:
\begin{proposition}\label{prop:convexity-one-way}
Let $\rho_0$ and $\rho_1$ be bipartite states on $AB$ satisfying
\begin{align}\label{eq:additivity-assumption}
\Do(\rho_{w^n}) \leq \sumi_i \Do(\rho_{w_i}) = (n-|w^n|) \Do(\rho_0) + |w^n| \Do(\rho_1)
\end{align}
for all $w^n\in\lbrace 0,1\rbrace^n$ and $n\in\mathbb{N}$.
Then for all $p\in [0,1]$,
\begin{align}
\Done(p\rho_0 + (1-p)\rho_1) \leq p \Done(\rho_0) + (1-p) \Done(\rho_1).
\end{align}
\end{proposition}

\begin{proof}
Let $n\in\mathbb{N}$, fix an instrument $T\colon A^n\to A'M$, and observe that we can write
\begin{align}
(p \rho_0 + (1-p)\rho_1)^{\ox n} = \sum_{w^n\in\lbrace 0,1\rbrace^n} p^{n-|w^n|} (1-p)^{|w^n|} \rho_{w^n}
\end{align}
using the notation introduced above.
Consider then the following steps:
\begin{align}
I(A'\rangle B^nM)_{T\left((p\rho_0 + (1-p)\rho_1)^{\ox n}\right)} &= I(A'\rangle B^nM)_{T\left(\sumi_{w^n} p^{n-|w^n|} (1-p)^{|w^n|} \rho_{w^n}\right)}\\
&= I(A'\rangle B^nM)_{\sumi_{w^n} p^{n-|w^n|} (1-p)^{|w^n|} T(\rho_{w^n})}\\
&\leq \sumi_{w^n} p^{n-|w^n|} (1-p)^{|w^n|} I(A'\rangle B^nM)_{ T(\rho_{w^n})},\label{eq:applied-convexity-of-coh-info}
\end{align}
where the second line follows from linearity of $T$, and in the last line we used convexity of the coherent information.
The latter in turn follows from joint convexity of the quantum relative entropy, defined for positive operators $\rho,\sigma$ with $\tr\rho=1$ as $D(\rho\|\sigma)\coloneqq \tr(\rho(\log\rho - \log\sigma))$ if $\supp\rho\subseteq\supp\sigma$, and set to $+\infty$ otherwise, and the fact that we can write $I(A\rangle B)_\tau = D(\tau_{AB}\|\one_A\ox\tau_B)$.
Maximizing both sides of \eqref{eq:applied-convexity-of-coh-info} over all instruments $T\colon A^n\to A'M$ and dividing by $n$, we obtain
\begin{align}
\frac{1}{n} \Do\left((p\rho_0 + (1-p)\rho_1)^{\ox n}\right) &\leq \frac{1}{n} \sumi_{w^n} p^{n-|w^n|} (1-p)^{|w^n|} \Do(\rho_{w^n})\\
&\leq \left[\frac{1}{n}\sumi_{w^n}(n-|w^n|)p^{n-|w^n|} (1-p)^{|w^n|}\right] \Do(\rho_0)\\
&\qquad  + \left[\frac{1}{n}\sumi_{w^n}|w^n| p^{n-|w^n|} (1-p)^{|w^n|}\right] \Do(\rho_1),\label{eq:separated}
\end{align}
where the last line follows from assumption \eqref{eq:additivity-assumption}.
Setting $j=|w^n|$, we have
\begin{align}
\frac{1}{n}\sumi_{w^n}|w^n| p^{n-|w^n|} (1-p)^{|w^n|} &= \frac{1}{n} \sum_{j=1}^n \binom{n}{j} j p^{n-j} (1-p)^j\\
&= \sum_{j=1}^{n} \binom{n-1}{j-1} p^{n-j} (1-p)^j\\
&= (1-p) \sum_{j=0}^{n-1} \binom{n-1}{j}p^{n-1-j}(1-p)^{j}\\
&= (1-p),
\end{align}
where we used the binomial identity $j\binom{n}{j}=n\binom{n-1}{j-1}$ in the second line, and the variable transformation $j\to j-1$ in the third line.
Similarly, we obtain 
\begin{align}
\frac{1}{n}\sumi_{w^n}(n-|w^n|)p^{n-|w^n|} (1-p)^{|w^n|} = p,
\end{align}
and taking the limit $n\to\infty$ in \eqref{eq:separated} yields
\begin{align}
\Done(p\rho_0 + (1-p)\rho_1) \leq p \Do(\rho_0) + (1-p)\Do(\rho_1).
\end{align}
The claim now follows from the fact that the subadditivity property \eqref{eq:additivity-assumption} implies $\Do(\rho_i^{\ox n}) = n\Do(\rho_i)$ for all $n\in\mathbb{N}$ and $i\in\lbrace 0,1\rbrace$, since we always have $\Do(\rho^{\ox n})\geq n \Do(\rho)$ for any arbitrary state $\rho$.
Hence, $\Done(\rho_i) = \Do(\rho_i)$ for $i\in\lbrace 0,1\rbrace$.
\end{proof}

If $\rho_0$ and $\rho_1$ are degradable, then the state $\rho_{w^n}$ is also degradable for any $w^n\in\lbrace 0,1\rbrace^n$ and $n\in\mathbb{N}$.
Hence, by \Cref{prop:dist-ent-single-letter} the assumption \eqref{eq:additivity-assumption} in \Cref{prop:convexity-one-way} is satisfied for degradable $\rho_0$ and $\rho_1$.
By \Cref{prop:forget-antidegradable}, the assumption \eqref{eq:additivity-assumption} is furthermore satisfied for tensor products of degradable and antidegradable states.

In summary, \Cref{prop:dist-ent-single-letter}, \Cref{prop:forget-antidegradable}, and \Cref{prop:convexity-one-way} prove that $\Done(\cdot)$ is convex on decompositions of an arbitrary bipartite state into degradable and antidegradable states.
Thus, we arrive at the upper bound advertised in \Cref{sec:main-results}, which we state in \Cref{thm:one-way-dist-upper-bound} below.
First, we recall the definition of the \emph{entanglement of formation} $E_F(\rho_{AB})$ of a bipartite state $\rho_{AB}$ \cite{BDSW96,BBP+96}:
\begin{align}\label{eq:eof}
E_F(\rho_{AB}) \coloneqq \min_{\lbrace p_i, \psi^i_{AB}\rbrace_i} \sumi_i p_i S(\psi^i_A),
\end{align}
where the minimization is over all pure-state ensembles $\lbrace p_i, \psi^i_{AB}\rbrace_i$ satisfying 
\begin{align}
\rho_{AB} = \sumi_i p_i |\psi^i\rangle\langle \psi^i|_{AB}.
\end{align}
\textcite{BDSW96} proved that $E_F(\rho_{AB})$ is an upper bound on the one-way distillable entanglement: $D_\rightarrow(\rho_{AB}) \leq E_F(\rho_{AB})$.
Our upper bound \Cref{thm:one-way-dist-upper-bound} below provides a refinement of this bound:

\begin{theorem}\label{thm:one-way-dist-upper-bound}
Let $\rho_{AB}$ be a bipartite state.
Then
\begin{align}
D_\rightarrow(\rho_{AB}) \leq \Eda(\rho_{AB}) \coloneqq \min \sum_{i=1}^k p_i I(A\rangle B)_{\rho_i} \leq E_F(\rho_{AB}),
\label{eq:D-arrow-bound}
\end{align}
where the minimization is over all decompositions of the form
\begin{align}\label{eq:deg-antideg-decomposition}
\rho_{AB} = \sum_{i=1}^k p_i \rho_i + \sum_{i=k+1}^l p_i \sigma_i
\end{align}
with degradable states $\rho_i$ and antidegradable states $\sigma_i$.
\end{theorem}
\begin{proof}
The first inequality follows from applying \Cref{prop:dist-ent-single-letter}, \Cref{prop:forget-antidegradable}, and \Cref{prop:convexity-one-way} to the decomposition of $\rho_{AB}$ in \eqref{eq:deg-antideg-decomposition}.
For the second inequality in \eqref{eq:D-arrow-bound}, recall that every pure state is degradable.
Hence, every decomposition of $\rho_{AB}$ into pure states is of the form \eqref{eq:deg-antideg-decomposition} (with $l=k$), in particular the one achieving the minimum in \eqref{eq:eof}.
\end{proof}

\subsection{2-qubit states and decompositions into degradable states}

In the case where both $A$ and $B$ are qubits, there is a simple method of obtaining decompositions of a bipartite state $\rho_{AB}$ into mixed degradable states.
This method is based on the following result by \textcite{WP07} about qubit-qubit quantum channels, which is easily extended to 2-qubit bipartite states:
\begin{proposition}[{\cite{WP07}}]\label{prop:rank-2}
Every qubit-qubit quantum channel with a qubit environment is either degradable or antidegradable.
Likewise, every 2-qubit bipartite state of rank 2 is either degradable or antidegradable.
\end{proposition}

\Cref{prop:rank-2} gives rise to an easy method for obtaining decompositions of a state $\rho_{AB}$ into mixed degradable and antidegradable states.
We first fix some $k\in\mathbb{N}$ such that $2k\geq \rk\rho_{AB}$, and decompose $\rho_{AB}$ into an even number $2k$ of \emph{pure} states:
\begin{align}
\rho_{AB} = \sum_{i=1}^{2k} p_i \psi_i.
\label{eq:tau-pure-decomposition}
\end{align} 
Note that every $2k\times 2k$ unitary matrix gives rise to such a pure-state decomposition \cite{HJW93}.
We then obtain rank-2 states from the pure states $\psi_i$ by grouping together two of them at a time:
for $j=1,\dots,k$, we set $q_j\coloneqq p_{2j-1} + p_{2j}$, and form the states
\begin{align}
\omega_j \coloneqq \frac{p_{2j-1}}{q_j} \psi_{2j-1} + \frac{p_{2j}}{q_j} \psi_{2j},
\end{align}
such that $\rho_{AB} = \sum_{j=1}^k q_j \omega_j$.
For every $j=1,\dots,k$ the state $\omega_j$ satisfies $\rk\omega_j=2$, and is therefore either degradable or antidegradable by \Cref{prop:rank-2}.
Hence, \Cref{thm:one-way-dist-upper-bound} yields the following upper bound on $D_\rightarrow(\rho_{AB})$:
\begin{align}
D_\rightarrow(\rho_{AB}) \leq \min_U \sum_{j\colon \omega_j\text{ deg.}} q_j I(A\rangle B)_{\omega_j},
\label{eq:tau-upper-bound}
\end{align}
where the minimization is over all $2k\times 2k$ unitary matrices $U$ determining the pure-state decomposition \eqref{eq:tau-pure-decomposition}, and the sum is over all $j$ such that $\omega_j$ is degradable.

\subsection{Approximate degradability}\label{sec:approx-degradability}

In \cite{SSWR15}, the authors introduced the concept of an approximate degradable quantum channel and used it to derive computable upper bounds on the quantum capacity of a given quantum channel.
More precisely, given a quantum channel $\cN$ and its complementary channel $\cN^c$, they defined the degradability parameter $\eps$ as the minimum distance in diamond norm between the complementary channel $\cN^c$ and a degraded version $\cD\circ\cN$ of the channel, minimized over all possible CPTP degrading maps $\cD$. 
That is, the degradability condition \eqref{eq:degradable-channel} for channels is only approximately satisfied in diamond norm up to $\eps$. 
The authors derived upper bounds on the quantum capacity $Q(\cN)$ (and the private capacity $P(\cN)$) of $\cN$ in terms of the channel coherent information of $\cN$ and error terms in $\eps$ that vanish in the limit $\eps\to 0$, hence reducing to the channel coherent information for degradable channels with $\eps=0$.
In this section, we formulate the notion of approximate degradable states in an analogous manner, using the trace distance between quantum states instead.
We then use similar ideas as in \cite{SSWR15} to derive an upper bound on the one-way distillable entanglement in terms of the coherent information and the degradability parameter.

For a bipartite quantum state $\rho_{AB}$ with purification $\phi_{ABE}$, the \emph{degradability parameter} $\dg(\rho_{AB})$ is defined as
\begin{align}
\dg(\rho_{AB}) \coloneqq \min_{\cD\colon B\to E} \frac{1}{2} \left\| \rho_{AE} - \cD(\rho_{AB}) \right\|_1\!,
\label{eq:approx-deg}
\end{align}
where $\rho_{AE} = \tr_B\phi_{ABE}$, the minimization is over CPTP maps $\cD\colon B\to E$, and the trace norm is defined as $\|X\|_1\coloneqq \tr\sqrt{X^\dagger X}$.
Similarly, we define the \emph{antidegradability parameter} $\adeg(\rho_{AB})$ as
\begin{align}
\adeg(\rho_{AB}) \coloneqq \min_{\cA\colon E\to B} \frac{1}{2} \left\| \rho_{AB} - \cA(\rho_{AE}) \right\|_1\!,
\label{eq:approx-antideg}
\end{align}
where the minimization is over CPTP maps $\cA\colon E\to B$.

The usefulness of the notion of $\eps$-degradable quantum channels stems from the fact that the degradability parameter $\eps$ can be formulated as the solution of a semidefinite program (SDP) \cite{SSWR15}, and is hence efficiently computable.
With our definition of the (anti-)degradability parameter in \eqref{eq:approx-deg} (resp.~\eqref{eq:approx-antideg}), this is also possible:

\begin{lemma}\label{lem:SDP}
	$\dg(\rho_{AB})$ is the solution of the SDP
	\begin{align}
	\begin{aligned}
	{\normalfont \text{minimize: }} & \frac{1}{4} (\tr X_{AE} + \tr Y_{AE})\\
	{\normalfont \text{subject to: }} & \begin{pmatrix}
	X_{AE} & Z_{AE} - \rho_{AE}\\
	Z_{AE} - \rho_{AE} & Y_{AE}
	\end{pmatrix} \geq 0\\
	& \tau_{B'E} \geq 0\\
	& \tau_{B'} = \one_B\\
	& X_{AE}, Y_{AE} \geq 0,
	\end{aligned}\label{eq:sdp-deg}
	\end{align}
	where $Z_{AE} = \tr_{B'}\left[\left(\rho_{AB'}^{T_B}\ox \one_E\right) \left(\one_A\ox\tau_{B'E}\right)\right]$ with $B'\cong B$ and $\rho_{AB'} = \rho_{AB}$, and where $\tau_{B'E}$ is the Choi state of the CPTP map $\cD\colon B\to E$ over which we optimize in \eqref{eq:approx-deg}.
	
	Similarly, $\adeg(\rho_{AB})$ is the solution of the SDP
	\begin{align}
	\begin{aligned}
	{\normalfont \text{minimize: }} & \frac{1}{4} (\tr X_{AB} + \tr Y_{AB})\\
	{\normalfont \text{subject to: }} & \begin{pmatrix}
	X_{AB} & W_{AB} - \rho_{AB}\\
	W_{AB} - \rho_{AB} & Y_{AB}
	\end{pmatrix} \geq 0\\
	& \tau_{E'B} \geq 0\\
	& \tau_{E'} = \one_{E'}\\
	& X_{AB}, Y_{AB} \geq 0,
	\end{aligned}\label{eq:sdp-antideg}
	\end{align}
	where $W_{AB} = \tr_{E'}\left[\left(\one_{B}\ox \rho_{AE}^{T_E}\right) \left(\one_A\ox\tau_{E'B}\right)\right]$ with $E'\cong E$ and $\rho_{AE'} = \rho_{AE}$, and where $\tau_{E'B}$ is the Choi state of the CPTP map $\cA\colon E\to B$ over which we optimize in \eqref{eq:approx-antideg}.
\end{lemma}

\begin{proof}
	Recall that for arbitrary $X\in\cB(\cH)$ the trace norm $\|X\|_1$ can be expressed as the following SDP (see e.g.~\cite[Ex.~1.15]{Wat16}):
	\begin{align}
	\begin{aligned}
	\text{minimize: } & \frac{1}{2}(\tr W_1 + \tr W_2)\\
	\text{subject to: } & \begin{pmatrix}
	W_1 & -X^\dagger \\
	-X & W_2
	\end{pmatrix} \geq 0,\\
	& W_1,W_2 \geq 0.
	\end{aligned}\label{eq:trace-norm-sdp}
	\end{align}
	The SDP formulations \eqref{eq:sdp-deg} and \eqref{eq:sdp-antideg} of \eqref{eq:approx-deg} and \eqref{eq:approx-antideg}, respectively, now follow immediately using the well-known Choi-Jamio\l kowski isomorphism.
\end{proof}

Based on ideas in \cite{SSWR15}, this notion of approximate (anti-)degradability allows us to derive a general, easily computable upper bound on the one-way distillable entanglement of an arbitrary bipartite state.
Before we state this result, we recall an improved version of the Alicki-Fannes inequality recently proved by \textcite{Win15}:

\begin{proposition}[{\cite{Win15}}]\label{prop:improved-alicki-fannes}
	Let $\rho_{AB}$ and $\sigma_{AB}$ be states with $\frac{1}{2}\|\rho_{AB}-\sigma_{AB}\|_1\leq \eps $, then
	\begin{align}
	|S(A|B)_\rho - S(A|B)_\sigma | \leq 2\eps\log|A| + \left(1+ \eps\right) h\!\left(\frac{\eps}{1+\eps}\right).
	\end{align}
\end{proposition}

\begin{theorem}\label{thm:approx-bound}
	Let $\rho_{AB}$ be a bipartite state with purification $|\phi\rangle_{ABE}$, and $\delta>0$ be such that $\dg(\rho_{AB})\leq \delta$.
	Then,
	\begin{align}
	I(A\rangle B)_\rho \leq D_\rightarrow(\rho_{AB}) \leq I(A\rangle B)_\rho + 4\delta \log|E| + 2\left(1 + \delta\right) h\!\left(\frac{\delta}{1+\delta}\right)\!,
	\end{align}
	where $h(\cdot)$ denotes the binary entropy.
\end{theorem}

\begin{proof}
	Let $\cD\colon B\to E$ be the CPTP map such that $\dg(\rho_{AB}) = \onehalf\|\rho_{AE} - \cD(\rho_{AB})\|_1 \leq \delta$, and denote by $W\colon B\to E'G$ its Stinespring isometry with $E'\cong E$.
	Consider the state $\rho_{AB}^{\ox n}$ and let $T\colon A^n\to A'M$ be an instrument with isometry 
	\begin{align}
	V_n\colon A^n\to A'MN,\quad V_n = \sumi_m U_m\ox |m\rangle_M \ox |m\rangle_N.
	\end{align}
	For $t=1,\dots,n$ we define the pure states
	\begin{align}
	|\psi^t\rangle_{A^n B_{t+1}\dots B_n E'_1 \dots E'_{t} G_1 \dots  G_{t} E_1\dots E_n} &= W_1\ox \dots \ox W_{t} |\phi\rangle_{ABE}^{\ox n}\\
	|\theta^t\rangle_{A'MN B_{t+1}\dots B_n E'_1 \dots E'_{t} G_1 \dots G_{t} E_1\dots E_n} &= V_n|\psi^t\rangle,
	\end{align}
	where $W_i = W \colon B_i\to E'_i G_i$ for every $i=1,\dots, t$. 
	Abbreviating $\theta=\theta^n$, we have the following:
	\begin{align}
	I(A'\rangle MB^n)_{V_n\rho^{\ox n}V_n^\dagger} &= I(A'\rangle MG^nE'^n)_\theta\\
	&= S(MG^nE'^n)_\theta - S(A'MG^nE'^n)_\theta\\
	&= S(MG^nE'^n)_\theta - S(NE^n)_\theta \label{eq:apply-conj-invariance-1}\\
	&= S(MG^nE'^n)_\theta - S(ME'^n)_\theta + S(ME'^n)_\theta - S(NE^n)_\theta\\
	&= S(G^n|ME'^n)_\theta + S(ME'^n)_\theta - S(ME^n)_\theta\\
	&= S(G^n|ME'^n)_\theta + \sum_{t=1}^n S(E'_t|M E'_{<t}E_{>t})_{\theta^t} - S(E_t|M E'_{<t}E_{>t})_{\theta^{t}}\label{eq:using-telescope}
	\end{align}
	where we used the symmetry of $\theta$ in $M$ and $N$ in the fifth equality, and the ``telescope'' identity \cite{LS08,SSWR15}
	\begin{align}\label{eq:telescope}
	S(ME'^n)_\theta - S(ME^n)_\theta = \sum_{t=1}^n S(E'_t|M E'_{<t}E_{>t})_{\theta^t} - S(E_t|M E'_{<t}E_{>t})_{\theta^{t}}
	\end{align}
	in the last equality, defining $X_{<t}\coloneqq X_{1}\dots X_{t-1}$, and setting $X_{<1}$ equal to a trivial (one-dimensional) system.
	$X_{>t}$ and $X_{>n}$ are defined analogously.
	The identity \eqref{eq:telescope} can be proved by simply writing out the right-hand side.
	
	For every $t=1,\dots,n$, we have the following bound on the trace distance between the two states $\theta^t_{ME'_1\dots E'_tE_{t+1}\dots E_n}$ and $\theta^t_{ME'_1\dots E'_{t-1}E_{t}\dots E_n}$ on which the coherent information is evaluated in \eqref{eq:using-telescope} resp.~\eqref{eq:telescope}:
	\begin{align}
	&\left\| \theta^t_{ME'_1\dots E'_tE_{t+1}\dots E_n} - \theta^t_{ME'_1\dots E'_{t-1}E_{t}\dots E_n} \right\|_1\\
	& \qqquad {}\leq \left\| \theta^t_{A'MNE'_1\dots E'_tE_{t+1}\dots E_n} - \theta^t_{A'MNE'_1\dots E'_{t-1}E_{t}\dots E_n} \right\|_1\\
	&\qqquad {}= \left\| \psi^t_{A^nE'_1\dots E'_tE_{t+1}\dots E_n} - \psi^t_{A^nE'_1\dots E'_{t-1}E_{t}\dots E_n} \right\|_1\\
	&\qqquad {}= \left\| \cD(\rho_{AB})^{\ox t} \ox \rho_{AE}^{\ox n-t} - \cD(\rho_{AB})^{\ox t-1} \ox \rho_{AE}^{\ox n-t+1} \right\|_1\\
	&\qqquad {} \leq \left \|\cD(\rho_{AB})^{\ox t-1} - \cD(\rho_{AB})^{\ox t-1} \right\|_1 + \left\| \cD(\rho_{AB}) - \rho_{AE}\right\|_1 + \left\|\rho_{AE}^{\ox n-t} - \rho_{AE}^{\ox n-t}\right\|_1\\
	&\qqquad {} \leq 2\delta,
	\end{align}
	where the second inequality follows from the fact that
	\begin{align}
	\|\rho_1\ox\rho_2 - \sigma_1\ox\sigma_2\|_1 \leq \|\rho_1-\rho_2\|_1 + \|\sigma_1 - \sigma_2\|_1
	\end{align}
	holds for any states $\rho_1,\rho_2,\sigma_1,\sigma_2$.
	Hence, by \Cref{prop:improved-alicki-fannes}, for every $t=1,\dots,n$ we have
	\begin{align}\label{eq:using-alicki-fannes}
	S(E'_t|M E'_{<t}E_{>t})_{\theta^t} - S(E_t|M E'_{<t}E_{>t})_{\theta^{t}} \leq 2\delta \log|E| + \left(1 + \delta\right) h\!\left(\frac{\delta}{1+\delta}\right)\eqqcolon \eps\!.
	\end{align}
	Using \eqref{eq:using-alicki-fannes} in \eqref{eq:using-telescope}, we then obtain
	\begin{align}
	I(A'\rangle MB^n)_{V_n\rho^{\ox n}V_n^\dagger} &\leq S(G^n|ME'^n)_\theta +  n \eps\\
	&\leq S(G^n|E'^n)_\theta +  n \eps\\
	&= S(G^nE'^n)_\theta - S(E'^n)_\theta + n\eps \label{eq:apply-conj-invariance-2}\\
	&\leq  S(G^nE'^n)_\theta - S(E^n)_\theta + 2 n\eps\\
	&= S(G^nE'^n)_\theta - S(A'MNG^nE'^n)_\theta + 2n\eps\\
	&= I(A'MN\rangle G^nE'^n)_\theta + 2n\eps\\
	&= I(A^n\rangle B^n)_{\rho^{\ox n}} + 2n\eps\\
	&= n (I(A\rangle B)_\rho + 2\eps),\label{eq:upper-bound-with-eps}
	\end{align}
	where in the third inequality we used a similar rewriting as in \eqref{eq:telescope} to bound the expression $S(E^n)_\theta-S(E'^n)_\theta$ from above by $n\eps$.
	The claim now follows after dividing \eqref{eq:upper-bound-with-eps} by $n$ and taking the limit $n\to\infty$.
\end{proof}

There is a generalized method of finding upper bounds on the one-way distillable entanglement that encompasses both the approximate degradability (AD) bound of this section and the `additive extension' (AE) bound in \Cref{thm:one-way-dist-upper-bound} in \Cref{sec:upper-bounds-one-way}.
As we will see later in \Cref{sec:symmetries-isotropic}, for the quantum capacity of the depolarizing channel the AD bound from \cite{SSWR15} provides the best upper bound for very low noise, while our AE bound does best for higher noise levels (cf.~\Cref{fig:depolarizing-bound}).  
By searching for approximately degradable extensions of quantum states (or channels, for that matter) we can do no worse than either of these two methods. 

The two methods can be combined as follows.
For a given bipartite state $\rho_{AB}$, fix $k\in\mathbb{N}$ and consider an extension $\tilde{\rho}_{ABC}$ with $|C|=k$, such that $\tr_C \tilde{\rho}_{ABC} = \rho_{AB}$.
We assume $C$ to be in Bob's possession, and consider entanglement distillation with respect to the $A|BC$ bipartition in the following.
Computing the degradability parameter $\eps = \dg(\tilde{\rho}_{ABC})$ of this extension and using \Cref{thm:approx-bound}, we obtain an upper bound on the one-way distillable entanglement $\Done(\tilde{\rho}_{ABC})$ of $\tilde{\rho}_{ABC}$, which in turn is an upper bound on $\Done(\rho_{AB})$.
We can then optimize this bound over all extensions $\tilde{\rho}_{ABC}$ with $|C|=k$.
Restricting to trivial extensions of $\rho_{AB}$, this bound reduces to the AD bound (\Cref{thm:approx-bound} in \Cref{sec:approx-degradability}).
Restricting to `flagged' (anti)degradable extensions of the form
\begin{align}
\tilde{\rho}_{ABC} = \sum_{c=1}^k \tilde{\rho}^c_{AB}\ox |c\rangle\langle c|_C,
\end{align}
where the states $\tilde{\rho}^c_{AB}$ are either degradable or antidegradable, the bound reduces to the AE bound (\Cref{thm:one-way-dist-upper-bound} in \Cref{sec:upper-bounds-one-way}).
In this case, we have $\rho_{AB} = \sum_c \tilde{\rho}^c_{AB}$.
It would be interesting to conduct a thorough numerical investigation of this approach.\footnote{We tried to implement this combined method to obtain upper bounds on the quantum capacity of the depolarizing channel (see \Cref{sec:symmetries-isotropic}). 
From prior numerical investigations, we know that the dimension $k$ of the extension register $C$ should be at least $6$.
However, for the choice $k=6$ the memory needed to solve the SDP in the computation of $\dg(\rho_{ABC})$ exceeds 96GB (even when exploiting the sparsity pattern of the Choi state of the depolarizing channel). Hence, such a computation is not tractable with the resources to which we have access.}

\subsection{Extending our method to the quantum capacity}

In this section we show that our method of obtaining an upper bound on the one-way distillable entanglement can be applied to quantum channels as well.
This allows us to easily establish upper bounds on the quantum capacity of a quantum channel of the form first reported by \textcite{Yang}.
We include our own argument for the result here for completeness, as it is a direct extension of the results in \Cref{sec:upper-bounds-one-way}.
Before explaining the main steps in the proof, we define the quantum capacity of a quantum channel in terms of the task of entanglement generation.

Let $\cN\colon A\to B$ be a quantum channel.
In entanglement generation, the goal for Alice (the sender) and Bob (the receiver) is to generate entanglement between them via $n$ uses of the channel $\cN$.
To this end, Alice prepares a pure state $|\phi\rangle_{A'A^n}$ in her laboratory and sends the $A^n$ part through the channel $\cN^{\ox n}$.
Bob then applies a decoding map $\cD\colon B^n\to \tA$ to the channel output state that he received from Alice.
The goal is to obtain a final state $(\id_{A'}\ox \cD\circ\cN^{\ox n})(\phi_{A'A^n})$ that is close to a maximally entangled state $\Phi_{A'\tA}^M$ of Schmidt rank $M$ up to some error $\eps_n$ (with respect to a suitable distance measure).
If there is an entanglement generation protocol for which $\lim_{n\to\infty}\eps_n = 0$, then $\lim_{n\to\infty}\frac{\log M}{n}$ is called an achievable rate for entanglement generation.
The quantum capacity $Q(\cN)$ is defined as the supremum over all achievable rates.

The following formula for the quantum capacity was proved (with increasing rigor) by \textcite{Llo97}, \textcite{Sho02}, and \textcite{Dev05}:
\begin{align}
Q(\cN) = \lim_{n \to\infty} \frac{1}{n} Q^{(1)}(\cN^{\ox n}),\label{eq:quantum-capacity}
\end{align}
where the \emph{channel coherent information} $Q^{(1)}(\cN)$ is defined as
\begin{align}
Q^{(1)}(\cN) \coloneqq \max_{|\phi\rangle_{A'A}} I(A'\rangle B)_{(\id\ox \cN)(\phi)}.
\end{align}
Similarly to the formula \eqref{eq:one-way-dist-entanglement} for the one-way distillable entanglement, the formula \eqref{eq:quantum-capacity} for the quantum capacity involves a regularization and is therefore intractable to compute in most cases.
However, much like their state counterparts for entanglement distillation, the classes of degradable and antidegradable channels that we defined in \Cref{sec:definitions} play a special role:
For a degradable quantum channel $\cN$, the channel coherent information is additive \cite{DS05},
\begin{align}
\Qone(\cN^{\ox n}) = n\Qone(\cN),
\end{align}
and thus the regularized formula \eqref{eq:quantum-capacity} reduces to the single-letter formula $Q(\cN) = \Qone(\cN)$.
Moreover, for antidegradable channels the channel coherent information, and hence the quantum capacity, is zero due to the no-cloning theorem \cite{BDSW96}.

Therefore, by once again using the ``additivity implies convexity'' argument by \textcite{WP07} (this time in its original form for quantum channels), we arrive at an upper bound to the quantum capacity, which is stated in \Cref{thm:quantum-capacity-upper-bound}.
This result is analogous to the upper bound for the one-way distillable entanglement in \Cref{thm:one-way-dist-upper-bound}.
The only missing piece is a channel analogue of \Cref{prop:forget-antidegradable}, which shows that an antidegradable channel does not contribute to the channel coherent information of a degradable channel.
This is a consequence of additivity of the channel coherent information for degradable channels \cite{DS05} and the technique of degradable extensions of a quantum channel \cite{SS08}.
Here, a quantum channel $\hat{\cN}$ is called \emph{extension} of a quantum channel $\cN$, if there is another quantum channel $\cR$ such that $\cN = \cR\circ\hat{\cN}$.

\begin{proposition}[\cite{DS05,SS08}]\label{prop:forget-antidegradable-channels}
	Let $\cN_1\colon A_1\to B_1$ be a degradable channel and $\cN_2\colon A_2\to B_2$ be an antidegradable channel.
	Then,
	\begin{align}
	\Qone(\cN_1\ox\cN_2) = \Qone(\cN_1).
	\end{align}
\end{proposition}

\begin{proof}
	It is proved in \cite{SS08} that for every antidegradable channel $\cA$ there is a degradable extension $\hat{\cA}$ of $\cA$ with vanishing quantum capacity, $Q(\hat{\cA})=0$.
	Let $\hat{\cN}_2$ be such a degradable extension for the antidegradable channel $\cN_2$.
	We then have the following:
	\begin{align}
	\Qone(\cN_1) \leq \Qone(\cN_1\ox\cN_2) \leq \Qone(\cN_1\ox \hat{\cN}_2) = \Qone(\cN_1) + \Qone(\hat{\cN}_2) = \Qone(\cN_1).
	\end{align}
	The first and second inequalities follow since $\cN_1\ox\cN_2$ and $\cN_1$ can be obtained from $\cN_1\ox\hat{\cN}_2$ and $\cN_1\ox\cN_2$ by post-processing, respectively.
	The first equality follows from additivity of $\Qone(\cdot)$ for degradable channels \cite{DS05}, and the second equality follows because $0\leq \Qone(\hat{\cN}_2) \leq Q(\hat{\cN}_2) = 0$.
	Hence, the above chain of inequalities collapses, which proves the claim.
\end{proof}

Finally, to arrive at our main result in this section we note that the proof of \Cref{prop:forget-antidegradable-channels} goes through if $\cN_1$ and $\cN_2$ are completely positive, but not necessarily trace-preserving.
Hence, we arrive at the following result:

\begin{theorem}[\cite{Yang}]\label{thm:quantum-capacity-upper-bound}
	For a quantum channel $\cN\colon A\to B$, 
	\begin{align}
	Q(\cN) \leq \min \sum_{i=1}^k p_i \Qone(\cD_i),
	\end{align}
	where the minimization is over all decompositions of the form
	\begin{align}
	\cN = \sum_{i=1}^k p_i \cD_i + \sum_{i=k+1}^l p_i \cA_i
	\end{align}
	with degradable and antidegradable CP maps $\cD_i$ and $\cA_i$, respectively.
\end{theorem}

\section{Two-way entanglement distillation}\label{sec:two-way}

\subsection{Operational setting}
In this section, we consider the task of entanglement distillation under two-way LOCC.
In contrast to the one-way setting, we do not concern ourselves with the structure of two-way LOCC operations.
Instead, we consider the larger class of $\ppt$-preserving operations, that is, the class of operations  $\Lambda\colon AB\to A'B'$ for which $\Lambda(\rho_{AB})^{\Gamma_{B'}} \geq 0$ whenever $\rho_{AB}^{\Gamma_{B}} \geq 0$.
Here, $\Gamma_B$ denotes transposition on the $B$ system.
We define the $\ppt$-distillable entanglement $D_\Gamma(\rho_{AB})$ in the same way as $\Done(\rho_{AB})$ or $\Dtwo(\rho_{AB})$, only this time with respect to $\ppt$-preserving operations.
Since every LOCC operation is also $\ppt$-preserving, we have 
\begin{align}
\Dtwo(\rho_{AB}) \leq D_\Gamma(\rho_{AB}). \label{eq:locc-ppt-dist-ent}
\end{align}

In the same vein as Rains' seminal work \cite{Rai99,Rai01}, we primarily derive upper bounds on the $\ppt$-distillable entanglement $D_\Gamma(\rho_{AB})$.
Subsequently, any such bound is also an upper bound on $\Dtwo(\rho_{AB})$ by \eqref{eq:locc-ppt-dist-ent}.

\subsection{Maximally correlated and PPT states}\label{sec:mcs}
Following the method outlined in \Cref{sec:main-results}, we first identify the classes of useful and useless states in the two-way LOCC and $\ppt$ setting.

\begin{definition}[\cite{Rai99}]\label{def:mcs}
	A bipartite state $\rho_{AB}$ on $\mathbb{C}^d\times\mathbb{C}^d$ is said to be \emph{maximally correlated} (MC), if there exist bases $\lbrace |i\rangle_A\rbrace_{i=0}^{d-1}$ and $\lbrace |i\rangle_B\rbrace_{i=0}^{d-1}$ such that
	\begin{align}
	\rho_{AB} = \sum_{i,j=0}^{d-1} \alpha_{ij} |i \rangle \langle j|_A \ox |i \rangle \langle j|_B,
	\end{align}
	where $(\alpha_{ij})$ is a positive semidefinite matrix with trace 1.
\end{definition}

Any pure state $|\psi\rangle_{AB}$ is MC, which can be seen by considering a Schmidt decomposition 
\begin{align}
|\psi\rangle_{AB} = \sumi_{i} \lambda_i |i\rangle_A\ox |i\rangle_B.
\end{align}
It then follows that $\psi_{AB}$ is MC with respect to the bases $\lbrace |i\rangle_A\rbrace_i$ and $\lbrace |i\rangle_B\rbrace_i$ and the matrix $(\alpha_{ij}) = \lambda_i \lambda_j$.

\begin{lemma}[\cite{Rai01}]
	For maximally correlated states $\rho_{AB}$,
	\begin{align}
	\Dtwo(\rho_{AB}) = I(A\rangle B)_\rho = I(B\rangle A)_\rho.
	\end{align}
	In particular, both $I(A\rangle B)_\rho$ and $I(B\rangle A)_\rho$ are non-negative for MC states.
\end{lemma}

\begin{proof}
	\textcite{Rai01} proved that for MC states $\rho_{AB}$ the $\ppt$-distillable entanglement $D_\Gamma(\rho_{AB})$ is equal to either one of the coherent informations, and thus
	\begin{align}
	\Dtwo(\rho_{AB}) \leq D_\Gamma(\rho_{AB}) = I(A\rangle B)_\rho = I(B\rangle A)_\rho.
	\end{align}
	On the other hand, by the hashing inequality \eqref{eq:hashing-bound} we have
	\begin{align}
	\Dtwo(\rho_{AB}) \geq \max\lbrace I(A\rangle B)_\rho, I(B\rangle A)_\rho\rbrace.
	\end{align}
	The non-negativity of the coherent informations of $\rho_{AB}$ now follows since they are equal to the operational quantity $\Dtwo(\rho_{AB})$.
	However, this can also be proved directly.	
	To this end, let $\rho_{AB} = \sum_{i,j=0}^{d-1} \alpha_{ij} |i \rangle \langle j|_A \ox |i \rangle \langle j|_B$ for suitable bases $\lbrace |i\rangle_A\rbrace_{i=0}^{d-1}$ and $\lbrace |i\rangle_B\rbrace_{i=0}^{d-1}$, and consider the projective measurement with measurement operators $P_k\coloneqq |k\rangle\langle k|_A\ox \one_B$.
	We have
	\begin{align}
	\omega_{AB} \coloneqq \sum_k P_k \rho_{AB} P_k = \sum_{i} \alpha_{ii} |ii\rangle\langle ii|_{AB},
	\end{align}
	and hence, $S(AB)_{\omega} = H(\lbrace \alpha_{11}, \dots, \alpha_{d-1,d-1}\rbrace) = S(B)_\rho$.
	Moreover, $S(AB)_\rho \leq S(AB)_\omega$, since projective measurements cannot decrease the von Neumann entropy. 
	It follows that $I(A\rangle B)_\rho = S(B)_\rho - S(AB)_\rho \geq 0$, and furthermore $I(A\rangle B)_\rho = I(B\rangle A)_\rho$.
\end{proof}

\begin{lemma}
	If $\rho_{AB}$ is $\ppt$, then $I(A\rangle B)_\rho \leq 0$.
\end{lemma}

\begin{proof}
	Clearly, $D_\Gamma(\rho_{AB}) = 0$ for all $\ppt$ states $\rho_{AB}$.
	Hence, 
	\begin{align} 
	0 = D_\Gamma(\rho_{AB}) \geq \Dtwo(\rho_{AB})\geq I(A\rangle B)_\rho,
	\end{align}
	where the last inequality follows from the hashing bound \eqref{eq:hashing-bound}.
\end{proof}

We now turn to the question of how to construct MC states.
We say that a collection of vectors $\lbrace |\psi_{\alpha}\rangle_{AB} \rbrace_{\alpha=1}^{l}$ of a bipartite quantum system with Hilbert space $\cH_A\ox\cH_B \cong \mathbb{C}^d\ox \mathbb{C}^d$ is \emph{simultaneously Schmidt decomposable} (SSD) \cite{HH04}, if there exist bases $\lbrace |i\rangle_A\rbrace_{i=0}^{d-1}$ and $\lbrace |i\rangle_B \rbrace_{i=0}^{d-1}$ of $\cH_A$ and $\cH_B$, respectively, such that
\begin{align}
|\psi_{\alpha} \rangle_{AB} = \sum_{i=0}^{d-1} \lambda_i^{(\alpha)} |i\rangle_A \ox |i\rangle_B\quad\text{for $\alpha=1,\dots,l$.}\label{eq:ssd}
\end{align}
In contrast to the usual Schmidt decomposition for a single bipartite pure quantum state, the coefficients $\lambda_i^{(\alpha)}$ are complex numbers in general.
It is clear by inspection of \eqref{eq:ssd} and \Cref{def:mcs} that, given a probability distribution $\lbrace p_\alpha\rbrace_{\alpha=1}^l$, the (mixed) state $\sum_{\alpha=1}^l p_\alpha |\psi_\alpha\rangle\langle \psi_\alpha|_{AB}$ is MC if the states $\lbrace |\psi_{\alpha}\rangle_{AB} \rbrace_{\alpha=1}^l$ are SSD.

We are therefore interested in necessary and sufficient conditions for a collection of vectors to be SSD.
By considering the associated operators $\lbrace \op(\psi_{\alpha})\rbrace_{\alpha=1}^{l}$ defined through \eqref{eq:vec-op}, this is equivalent to the existence of a \emph{weak singular value decomposition} for $\lbrace \op(\psi_{\alpha})\rbrace_{\alpha=1}^{l}$, by which we mean that there are unitary matrices $U$ and $V$ such that the matrices $U \op(\psi_\alpha) V$ are (complex) diagonal for all $\alpha=1,\dots,l$.
Necessary and sufficient conditions for the existence of such weak singular value decompositions for a set $\lbrace A_i\rbrace_i$ of matrices were found by \textcite{Wie48} and further refined by \textcite{Gib74}.
In our context, their results can be phrased as follows:
\begin{theorem}[\cite{Wie48,Gib74}] \label{thm:ssd}
	For quantum systems $A$ and $B$, let $\lbrace |\psi_{\alpha}\rangle_{AB} \rbrace_{\alpha=1}^{l}$ be a collection of vectors, and let $\mathcal{S} = \lbrace \op(\psi_\alpha)\rbrace_{\alpha=1}^{l}$ be the set of associated operators.
	Then $\lbrace |\psi_{\alpha}\rangle_{AB} \rbrace_{\alpha=1}^{l}$ is SSD if and only if 
	\begin{align}
	X Y^\dagger Z = Z Y^\dagger X\quad\text{for all $X,Y,Z\in\mathcal{S}$.} \label{eq:ssd-condition}
	\end{align}
\end{theorem}
The concept of simultaneous Schmidt decomposition was introduced in quantum information theory  by \textcite{HH04}, who proved an alternative version of \Cref{thm:ssd}.

An easy consequence of \Cref{thm:ssd} is the following
\begin{lemma}
	Let $U$ be a unitary on $\mathbb{C}^d$.
	Then $(\one_d \ox U)|\Phi_+\rangle$ and $(\one_d \ox U^\dagger)|\Phi_+\rangle$ are SSD.
	Moreover, the set $\lbrace (\one_d \ox U^i)|\Phi_+\rangle\rbrace_{i=0}^{d-1}$ is SSD.
\end{lemma}

\begin{proof}
	Setting $|\psi_1\rangle = (\one_d \ox U)|\Phi_+\rangle$ and $|\psi_2\rangle = (\one_d \ox U^\dagger)|\Phi_+\rangle$, we have $\op(\psi_1) = \frac{1}{\sqrt{d}} U^\dagger$ and $\op(\psi_2) = \frac{1}{\sqrt{d}} U$.
	Since $[U,U^\dagger] = 0$ holds for any unitary, the condition of \Cref{thm:ssd} is satisfied for all choices of $X,Y,Z$. 
	The same argument shows that the set $\lbrace (\one_d \ox U^i)|\Phi_+\rangle\rbrace_{i=0}^{d-1}$ is also SSD. 
\end{proof}

\subsection{Upper bounds on the two-way distillable entanglement}\label{sec:upper-bounds-two-way}

To prove the main result of this section, we define the relative entropy of entanglement $E_R^X(\rho_{AB})$ for $X\in\lbrace \ppt,\sep \rbrace$ as
\begin{align}
E_R^X(\rho_{AB}) \coloneqq \min_{\sigma_{AB}\in X} D(\rho_{AB}\|\sigma_{AB}),\label{eq:relative-entropy-of-entanglement}
\end{align}
where $\sep$ and $\ppt$ denote the sets of bipartite separable and PPT states on $\cH_A\ox\cH_B$, respectively.

\begin{proposition}\label{prop:convexity-two-way}
	The two-way distillable entanglement is convex on convex combinations of $\mcs$ and $\ppt$ states.
\end{proposition}

\begin{proof}
	First, recall that $D_\Gamma(\rho_{AB}) \leq E_R^\ppt(\rho_{AB})$ \cite{Rai99}.
	Consider now the following decomposition of a state $\rho_{AB}$,
	\begin{align}
	\rho_{AB} = \sum_{i=1}^k p_i \omega_i + \sum_{i=k+1}^l p_i \tau_i,\label{eq:MCS-PPT-decomposition}
	\end{align}
	where the $\omega_i$ are MC states and the $\tau_i$ are PPT.
	Since $\omega_i$ is MC, there are bases $\lbrace |i\rangle_{A}\rbrace_i$ and $\lbrace |i\rangle_{B}\rbrace_i$ of $\cH_A$ and $\cH_B$, respectively, such that $\omega_i = \sum_{k,l} \alpha_{kl} |kk\rangle\langle ll|_{AB}$.
	The dephased state $\omega_i' = \sum_k \alpha_{kk} |kk\rangle \langle kk|_{AB}$ is manifestly PPT, and satisfies  \cite{Rai99}
	\begin{align}
	D(\omega_i\|\omega_i') = I(A\rangle B)_{\omega_i} = \Dtwo(\omega_i).
	\end{align}
	Since the set of PPT states is convex, the state $\sigma_{AB}\coloneqq \sum_{i=1}^k p_i \omega_i' + \sum_{i=k+1}^l p_i \tau_i$ is also PPT, and we obtain the following chain of inequalities:
	\begin{align}
	\Dtwo(\rho_{AB}) &\leq D_\Gamma(\rho_{AB})\\
	&\leq E_R^{\ppt}(\rho_{AB})\\
	&\leq D(\rho_{AB}\|\sigma_{AB})\\
	&\leq \sum_{i=1}^k p_i D(\omega_i\| \omega_i') + \sum_{i=k+1}^l p_i D(\tau_i \| \tau_i)\\
	&= \sum_{i=1}^k p_i \Dtwo(\omega_i) + \sum_{i=k+1}^l p_i \Dtwo(\tau_i)\\
	&= \sum_{i=1}^k p_i I(A\rangle B)_{\omega_i},\label{eq:thm-rhs}
	\end{align}
	where we used joint convexity of $D(\cdot\|\cdot)$ in the last inequality.
\end{proof}

We can now formulate our main result:

\begin{theorem}\label{thm:two-way-dist-upper-bound}
	For a bipartite state $\rho_{AB}$, we have the following upper bound on the two-way distillable entanglement:
	\begin{align}
	\Dtwo(\rho_{AB}) \leq \Emp(\rho_{AB})\coloneqq \min \sum_{i=1}^k p_i I(A\rangle B)_{\omega_i} \leq E_F(\rho_{AB}),\label{eq:E_MP}
	\end{align}
	where the minimization is over all decompositions of $\rho_{AB}$ of the form
	\begin{align}
	\rho_{AB} = \sum_{i=1}^k p_i \omega_i + \sum_{i=k+1}^l p_i \tau_i,
	\end{align}
	and the $\omega_i$ and $\tau_i$ are MC and PPT states, respectively.
\end{theorem}

\begin{proof}
	The first inequality immediately follows from \Cref{prop:convexity-two-way}, minimizing over all decompositions of $\rho_{AB}$ of the form \eqref{eq:MCS-PPT-decomposition}.	
	The second inequality follows from the fact that every pure state $\psi_{AB}$ is MC.
\end{proof}

We also define 
\begin{align}
\Em(\rho_{AB}) \coloneqq \min \sum_{i=1}^k p_i I(A\rangle B)_{\omega_i},
\end{align}
where the minimization is now restricted to decompositions of $\rho_{AB}$ into MC states alone, that is, decompositions of the form
\begin{align}
\rho_{AB} = \sum_{i=1}^k p_i \omega_i,
\end{align}
where the $\omega_i$ are MC.
Clearly, $\Emp(\rho_{AB}) \leq \Em(\rho_{AB})$ for all $\rho_{AB}$.

\begin{lemma}~\label{lem:rel-ent}
	\begin{enumerate}[\normalfont (i)]
		\item\label{item:ppt-ree} For all $\rho_{AB}$, $E_R^{\ppt}(\rho_{AB}) \leq \Emp(\rho_{AB})$.
		
		\item\label{item:sep-ree} For all $\rho_{AB}$, $E_R^{\sep}(\rho_{AB}) \leq \Em(\rho_{AB})$.
		
		\item\label{item:mcs-em} If $\rho_{AB}$ is MC, then $\Em(\rho_{AB}) = I(A\rangle B)_\rho = I(B\rangle A)_\rho = E_R^\sep(\rho_{AB})$.
		
		\item\label{item:emp-sep-gap} There are states $\rho_{AB}$ for which $\Emp(\rho_{AB}) < E_R^\sep(\rho_{AB})$.
	\end{enumerate}
\end{lemma}

\begin{proof}
	\Cref{item:ppt-ree} is clear from the proof of \Cref{thm:two-way-dist-upper-bound}.
	The same line of arguments for a decomposition $\rho_{AB} = \sum_{i=1}^k p_i \omega_i$ into MC states alone, together with the fact that $E_R^\ppt(\rho_{AB})\leq E_R^\sep(\rho_{AB})$ for all $\rho_{AB}$, shows \Cref{item:sep-ree}.
	
	\Cref{item:mcs-em} Let $\rho_{AB}=\sum p_i \omega_i$ be a further decomposition of the MC state $\rho_{AB}$ into MC states $\omega_i$.
	Then
	\begin{align}
	I(A\rangle B)_\rho \leq \sumi_i p_i I(A\rangle B)_{\omega_i}
	\end{align}
	by the convexity of the coherent information.
	Hence, the trivial decomposition of $\rho_{AB}$ into MC states achieves a minimum among all such decompositions, and hence $\Em(\rho_{AB})= I(A \rangle B)_\rho$.
	
	To prove that for an MC state $\rho_{AB}$ also $E_R^\sep(\rho_{AB}) = I(A\rangle B)_\rho$, we note that for general states $\rho_{AB}$ we have \cite{PVP00}
	\begin{align}
	E_R^\sep(\rho_{AB}) \geq \max \lbrace I(A\rangle B)_\rho, I(B\rangle A)_\rho,0\rbrace.
	\end{align}
	Together with \Cref{item:sep-ree}, this implies for an MC state $\rho_{AB}$ that
	\begin{align}
	I(A\rangle B)_\rho = \Em(\rho_{AB}) \geq E_R^\sep(\rho_{AB}) \geq I(A\rangle B)_\rho.
	\end{align}
	Hence, this chain of inequalities collapses, and the same holds for the one with $I(B\rangle A)_\rho$.
	
	To prove \Cref{item:emp-sep-gap}, note that any entangled PPT (and hence bound entangled) state $\rho_{AB}$  satisfies $0 = \Emp(\rho_{AB}) < E_R^\sep(\rho_{AB})$.
\end{proof}

In view of \Cref{lem:rel-ent}, an interesting question is whether
\begin{align}
\Emp(\rho_{AB}) \overset{?}{\leq} E_R^\sep(\rho_{AB})
\end{align}
holds for all $\rho_{AB}$.

\subsection{Block-diagonal states in the generalized Bell basis}\label{sec:blocking}

In this section, we investigate our upper bound $\Emp(\cdot)$ on quantum states that are block-diagonal in the generalized Bell basis $\lbrace |\Phi_{n,m}\rangle \rbrace_{n,m=0,\dots,d-1}$, where $|\Phi_{n,m}\rangle$ is defined by \eqref{eq:bell-states}.
To this end, we first note that the SSD criterion for pure states given in \Cref{thm:ssd} reduces to a simple algebraic relation for the generalized Bell basis (see also \cite{HH04} for a similar relation):
\begin{corollary}\label{cor:algebraic-relation}
	A subset $\lbrace |\Phi_{n_\alpha,m_\alpha}\rangle \rbrace_{\alpha=1,\dots,l}$ of the generalized Bell states with $l\leq d$ is SSD if and only if the following equation is satisfied for all $\alpha,\beta,\gamma\in [l]$:
	\begin{align}
	m_\alpha(n_\gamma - n_\beta) - n_\gamma m_\beta = n_\alpha(m_\gamma - m_\beta) - m_\gamma n_\beta\mod d
	\label{eq:congruence-relations}
	\end{align}
\end{corollary}

Note that the rank of an MC state $\rho_{AB} = \sum_{i,j} \alpha_{ij} |i\rangle\langle j|_A\ox |i\rangle\langle j|_B$ is equal to the rank of the $|A|\times |B|$ matrix $(\alpha_{ij})$, and hence at most $\min\lbrace|A|,|B|\rbrace$.
In this section, $|A|=|B|=d$, and we focus on MC states with maximal rank $d$ that lie in the span of a collection of $d$ distinct Bell states $\lbrace |\Phi_{n_\alpha,m_\alpha}\rangle \rbrace_{\alpha=1,\dots,d}$.
Let us first introduce a different numbering $k\equiv k(n,m) = nd + m + 1$ for the generalized Bell states, and take $d=3$ and $B=\lbrace 1,6, 8\rbrace$ as an example.
Then \Cref{cor:algebraic-relation} implies that any state of the form
\begin{align}
\omega_{AB} = \sum_{i,\,j\in B} \alpha_{ij} |\Phi_i\rangle \langle \Phi_j|_{AB} \label{eq:d3-mc-state}
\end{align}
with $(\alpha)_{ij}\geq 0$ and $\tr\alpha = 1$ is MC.
In fact, we can use \Cref{cor:algebraic-relation} to search for all blocks $B$ of size $d$ such that a state in the span of $\lbrace |\Phi_{i}\rangle\rbrace_{i\in B}$ is MC.
\Cref{tab:partitions} lists all these blocks of size $d$ for $d\in \lbrace 2,3,4,5\rbrace$.

\begin{table}[ht]
	\centering
	\begin{tabular}{r l}
		\toprule
		$d$ & blocks\\
		\midrule
		$2$ & $\lbrace 1, 2 \rbrace$, $\lbrace 1, 3\rbrace$, $\lbrace 1, 4\rbrace$, $\lbrace 2, 3\rbrace$, $\lbrace 2,4\rbrace$, $\lbrace 3, 4\rbrace$\\
		\midrule
		$3$ & $\lbrace 1  ,   2   ,  3 \rbrace$, 
		$\lbrace 1  ,   4  ,   7\rbrace$, 
		$\lbrace 1  ,   5  ,   9\rbrace$, 
		$\lbrace 1  ,   6  ,   8\rbrace$, 
		$\lbrace 2  ,   4  ,   9\rbrace$, 
		$\lbrace 2  ,   5  ,   8\rbrace$, 
		$\lbrace 	2  ,   6 ,    7\rbrace$,\\
		& $\lbrace 	3   ,  4   ,  8\rbrace$,	
		$\lbrace 3   ,  5  ,   7\rbrace$, 
		$\lbrace 	3  ,   6   ,  9\rbrace$, 
		$\lbrace 	4  ,   5   ,  6\rbrace$, 
		$\lbrace 	7  ,   8   ,  9\rbrace$\\
		\midrule
		$4$ &  $\lbrace  1  ,   2  ,   3   ,  4\rbrace$, 
		$\lbrace 1  ,   3  ,   9  ,  11\rbrace$, 
		$\lbrace 1  ,   3  ,  10  ,  12\rbrace$,
		$\lbrace 1  ,   5  ,   9  ,  13\rbrace$,
		$\lbrace 1  ,   6  ,  11  ,  16\rbrace$,\\
		& $\lbrace 1  ,   7  ,   9  ,  15\rbrace$,
		$\lbrace 1  ,   8  ,  11  ,  14\rbrace$, 
		$\lbrace 2  ,   4  ,   9  ,  11\rbrace$,
		$\lbrace 2  ,   4  ,  10  ,  12\rbrace$,
		$\lbrace 2  ,   5  , 12  ,  15\rbrace$,\\
		& $\lbrace 2  ,   6  ,  10  ,  14\rbrace$,
		$\lbrace 2  ,   7  ,  12  ,  13\rbrace$,
		$\lbrace 2  ,   8  ,  10  ,  16\rbrace$, 
		$\lbrace 3  ,   5  ,  11  ,  13\rbrace$,
		$\lbrace 3  ,   6  ,   9  ,  16\rbrace$,\\
		& $\lbrace 3  ,   7  ,  11  ,  15\rbrace$,
		$\lbrace 3  ,   8  ,   9  ,  14\rbrace$, 
		$\lbrace 4  ,   5  ,  10  ,  15\rbrace$,
		$\lbrace 4  ,   6  ,  12  ,  14\rbrace$,
		$\lbrace 4  ,   7  ,  10  ,  13\rbrace$,\\
		& $\lbrace 4  ,   8  ,  12  ,  16\rbrace$, 
		$\lbrace 5   ,  6  ,   7   ,  8\rbrace$,
		$\lbrace 5  ,   7  ,  13  ,  15\rbrace$,
		$\lbrace 5   ,  7  ,  14  ,  16\rbrace$, 
		$\lbrace 6  ,   8  ,  13  ,  15\rbrace$,\\
		& $\lbrace 6   ,  8  ,  14  ,  16\rbrace$, 
		$\lbrace 9  ,  10  ,  11  ,  12\rbrace$, 
		$\lbrace 13  ,  14  ,  15 , 16\rbrace$ \\
		\midrule
		$5$ & $\lbrace 1  ,   2  ,   3   ,  4   ,  5\rbrace$,
		$\lbrace 1  ,   6  ,  11  ,  16  ,  21\rbrace$,
		$\lbrace 1  ,   7  ,  13  ,  19  ,  25\rbrace$,
		$\lbrace 1  ,   8  ,  15  ,  17  ,  24\rbrace$,\\
		& $\lbrace 1  ,   9  ,  12  ,  20  ,  23\rbrace$, 
		$\lbrace 1  ,  10  ,  14  ,  18  ,  22\rbrace$, 
		$\lbrace 2  ,   6   , 15  ,  19  ,  23\rbrace$,
		$\lbrace 2   ,  7  ,  12  ,  17  ,  22\rbrace$,\\
		& $\lbrace 2  ,   8  ,  14  ,  20  ,  21\rbrace$,
		$\lbrace 2  ,   9  ,  11  ,  18  ,  25\rbrace$,
		$\lbrace 2  ,  10  ,  13  ,  16  ,  24\rbrace$, 
		$\lbrace 3  ,   6  ,  14  ,  17  ,  25\rbrace$,\\
		& $\lbrace 3  ,   7  ,  11 ,   20  ,  24\rbrace$,
		$\lbrace 3  ,   8  ,  13  ,  18  ,  23\rbrace$,
		$\lbrace 3  ,   9  ,  15  ,  16  ,  22\rbrace$,
		$\lbrace 3  ,  10  ,  12  ,  19  ,  21\rbrace$, \\
	& 	$\lbrace 4  ,   6  ,  13  ,  20  ,  22\rbrace$,
		$\lbrace 4  ,   7  ,  15  ,  18  ,  21\rbrace$,
		$\lbrace 4   ,  8  ,  12  ,  16  ,  25\rbrace$,
		$\lbrace 4  ,   9  ,  14  ,  19  ,  24\rbrace$,\\
		& $\lbrace 4  ,  10  ,  11  ,  17  ,  23\rbrace$, 
		$\lbrace 5  ,   6  ,  12  ,  18  ,  24\rbrace$,
		$\lbrace 5  ,   7  ,  14  ,  16 ,   23\rbrace$,
		$\lbrace 5  ,   8  ,  11  ,  19 ,   22\rbrace$,\\
		& 	$\lbrace 5  ,   9  ,  13  ,  17  ,  21\rbrace$,
		$\lbrace 5  ,  10  ,  15  ,  20  ,  25\rbrace$, 
		$\lbrace 6  ,   7   ,  8  ,   9  ,  10\rbrace$, 
		$\lbrace 11  ,  12  ,  13   , 14  ,  15\rbrace$, \\
		& $\lbrace 16  ,  17  ,  18  ,  19  ,  20\rbrace$, 
		$\lbrace 21  ,  22  ,  23  ,  24  ,  25\rbrace$\\
		\bottomrule
	\end{tabular}
	\caption{Blocks of size $d$ for $d\in \lbrace 2,3,4,5\rbrace$ giving rise to MC states according to \Cref{cor:algebraic-relation}.}
	\label{tab:partitions}
\end{table}

\Cref{cor:algebraic-relation} and \Cref{tab:partitions} provide a method of constructing MC states that are block-diagonal in the generalized Bell basis, allowing us to test the quality of our upper bound $\Emp(\cdot)$ on $\Dtwo(\cdot)$.
As a benchmark we use the following SDP bound on $\Dtwo(\cdot)$ recently derived by \textcite{WD16}:
\begin{align}
\Ewd(\rho_{AB}) &\coloneqq \log \left[ \max \left\lbrace \tr(\rho_{AB} R_{AB}) \colon R_{AB}\geq 0, -\one_{AB} \leq R_{AB}^{\Gamma_B} \geq \one_{AB} \right\rbrace \right]\!.  \label{eq:wang-duan}
\end{align}
In \cite{WD16} the authors proved that $\Ewd(\rho_{AB}) \leq E_N(\rho_{AB})$ for all states $\rho_{AB}$, where 
\begin{align}
E_N(\rho_{AB}) = \log \|\rho_{AB}^{\Gamma_B}\|_1
\end{align} 
is the logarithmic negativity \cite{VW02,Ple05}.
We set $d=3$ and consider states of the form
\begin{align}
\rho_{AB} = (1-p)\, \omega_{AB} + p\, \tau_{AB},\label{eq:approx-mc-state}
\end{align}
where $p\in[0,1]$, the state $\omega_{AB}$ is defined as in \eqref{eq:d3-mc-state} for a valid block $B$ of size $3$ from \Cref{tab:partitions}, and $\tau_{AB}$ is the following PPT entangled state with $a=\frac{1}{2}$ \cite{Hor97}:
\begin{align}
\tau_{AB} = \frac{1}{8a+1} \begin{pmatrix}
a    &     0    &     0    &    0  &  a    &     0    &     0    &     0  &  a \\
0  &  a     &    0    &     0    &     0   &      0   &      0    &     0    &     0 \\
0   &      0  &  a    &     0    &     0    &     0    &     0    &     0    &     0 \\
0    &     0   &      0  &  a    &     0    &     0   &      0    &     0    &     0 \\
a    &     0    &     0    &     0  &  a    &     0    &     0    &     0  &  a \\
0    &     0    &     0   &      0   &      0  &  a    &     0    &     0    &     0 \\
0    &     0    &     0    &     0    &     0   &      0  &  \frac{1+a}{2}   &      0  &  \frac{\sqrt{1-a^2}}{2} \\
0    &     0    &     0    &     0    &     0    &     0    &     0  &  a    &     0 \\
a    &     0    &     0    &     0  &  a    &     0  &  \frac{\sqrt{1-a^2}}{2}    &     0  &  \frac{1+a}{2}
\end{pmatrix}
\end{align}

In \Cref{fig:d3-block}, we compare the bounds $\Emp(\cdot)$ and $\Ewd(\cdot)$ for $1000$ random states of the form given in \eqref{eq:approx-mc-state} (selecting both the block $B$ as well as the matrix $\alpha$ uniformly at random) for the values $p\in\lbrace 0.1,0.25,0.5,0.75\rbrace$.
For a state $\rho_{AB}$ of the form given in \eqref{eq:approx-mc-state}, our bound evaluates to
\begin{align}
\Emp(\rho_{AB}) = (1-p) I(A\rangle B)_\omega
\end{align} 
due to \Cref{thm:two-way-dist-upper-bound}.
Evidently, it performs particularly well for low values of $p$, for which the state $\rho_{AB}$ is almost MC.

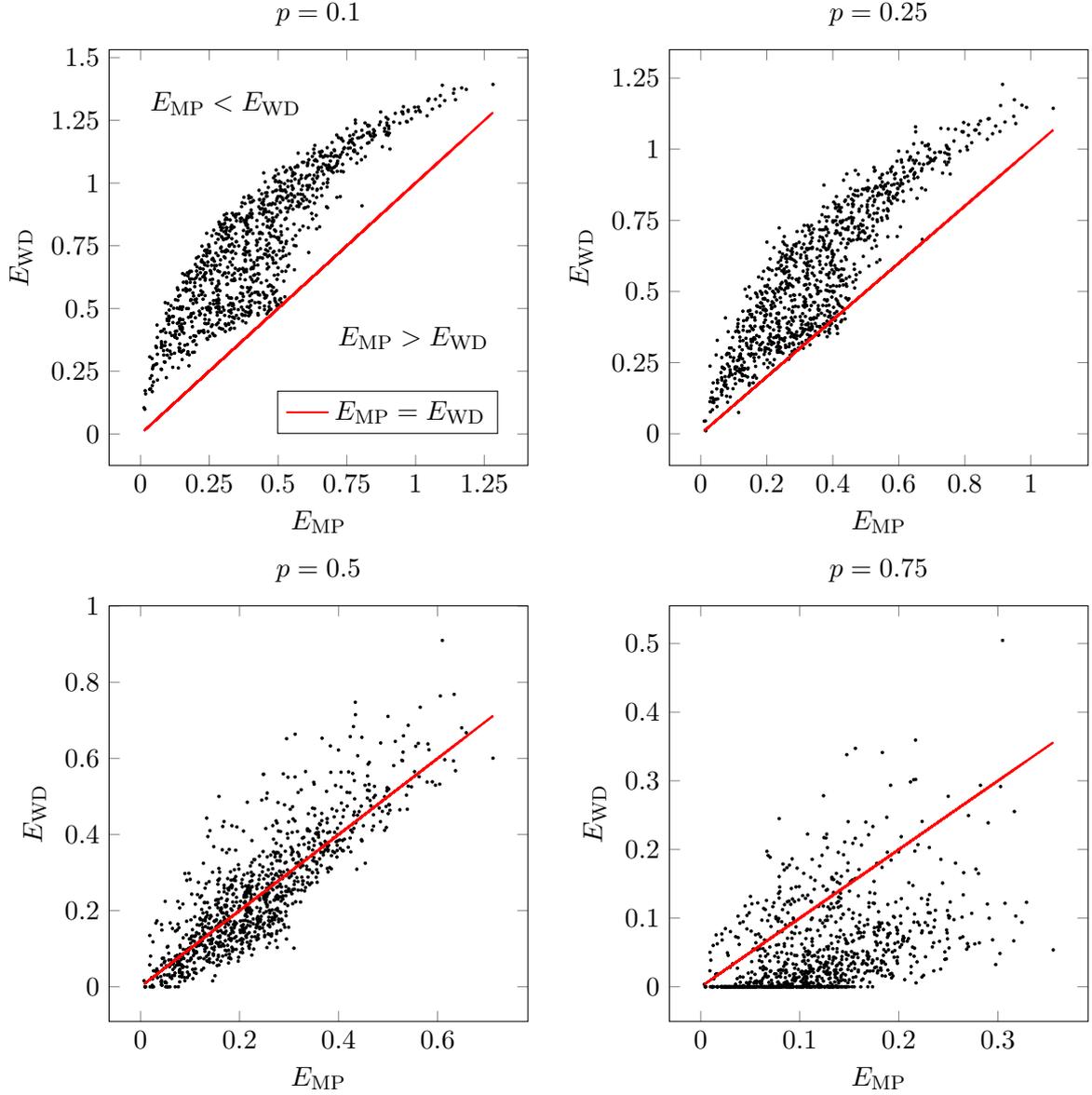
\begin{figure}
	\centering
	\begin{tikzpicture}
	\tikzset{mark size=0.5}
	\begin{axis}[name=plot1,title={$p=0.1$},height=7.5cm,width=7.5cm,clip mode=individual,xlabel = $\Emp$,ylabel=$\Ewd$,legend style = {at = {(axis cs:1.3,0)},anchor = south east},xtick = {0,0.25,0.5,0.75,1,1.25},ytick = {0,0.25,0.5,0.75,1,1.25,1.5}]
	\addplot[only marks,color=black] table[x=u,y=w] {p010.dat};
	\addplot[color=red,thick] table[x=u,y=u] {p010.dat};
	\node[anchor = north west] at (axis cs:0,1.4) {$\Emp < \Ewd$};
	\node[anchor = south east] at (axis cs:1.3,0.3) {$\Emp > \Ewd$};
	\legend{,$\Emp = \Ewd$};
	\end{axis}
	\begin{axis}[name=plot2,title={$p=0.25$},at={($(plot1.east)+(2cm,0)$)},anchor=west,height=7.5cm,width=7.5cm,clip mode=individual,xlabel = $\Emp$,ylabel=$\Ewd$,ytick = {0,0.25,0.5,0.75,1,1.25}]
	\addplot[only marks,color=black] table[x=u,y=w] {p025.dat};
	\addplot[color=red,thick] table[x=u,y=u] {p025.dat};
	\end{axis}
	\begin{axis}[name=plot3,title={$p=0.5$},at={($(plot1.south)-(0,2cm)$)},anchor=north,height=7.5cm,width=7.5cm,clip mode=individual,xlabel = $\Emp$,ylabel=$\Ewd$]
	\addplot[only marks,color=black] table[x=u,y=w] {p050.dat};
	\addplot[color=red,thick] table[x=u,y=u] {p050.dat};
	\end{axis}
	\begin{axis}[name=plot4,title={$p=0.75$},at={($(plot2.south)-(0,2cm)$)},anchor=north,height=7.5cm,width=7.5cm,clip mode=individual,xlabel = $\Emp$,ylabel=$\Ewd$,ytick = {0,0.1,0.2,0.3,0.4,0.5}]
	\addplot[only marks,color=black] table[x=u,y=w] {p075.dat};
	\addplot[color=red,thick] table[x=u,y=u] {p075.dat};
	\end{axis}
	\end{tikzpicture}
	\caption{Plot of $\Ewd(\rho_{AB})$ \cite{WD16} given in \eqref{eq:wang-duan} against $\Emp(\rho_{AB})$ from \Cref{thm:two-way-dist-upper-bound}, with each dot corresponding to one of 1000 randomly generated (according to the Haar measure) states $\rho_{AB}$ as defined in \eqref{eq:approx-mc-state} with the indicated value of $p\in \lbrace 0.1,0.25,0.5,0.75\rbrace$ from top left to bottom right, respectively. The red line indicates that $\Emp(\cdot) = \Ewd(\cdot)$, and a dot above (resp.~below) the red line indicates a state $\rho_{AB}$ for which $\Emp(\rho_{AB}) < \Ewd(\rho_{AB})$ (resp.~$\Emp(\rho_{AB}) > \Ewd(\rho_{AB})$).}
	\label{fig:d3-block}
\end{figure}

Particular examples of states of the form as in \eqref{eq:approx-mc-state} are
\begin{align}
\theta_{AB}^{(k)} = (1-p) \sum_{i,\,j\in \lbrace 1, 6, 8\rbrace} \alpha_{ij}^{(k)} |\Phi_i\rangle \langle \Phi_j|_{AB} + p\,\tau_{AB} \label{eq:theta}
\end{align}
for $k=1,2$, where
\begin{align}
\alpha^{(1)} &= \frac{1}{2} \left( |0\rangle\langle 0| + |\psi\rangle\langle\psi| \right) &
\alpha^{(2)} &= |\psi\rangle\langle\psi| &
|\psi\rangle &= \frac{1}{\sqrt{3}}\left( |0\rangle + |1\rangle + |2\rangle \right)\!,
\end{align}
and where $\lbrace |0\rangle, |1\rangle, |2\rangle\rbrace$ is the computational basis of $\mathbb{C}^3$.
In \Cref{fig:d3-example}, we plot $\Emp(\theta_{AB}^{(k)})$ and $\Ewd(\theta_{AB}^{(k)})$ for $k=1,2$ as a function of $p$.

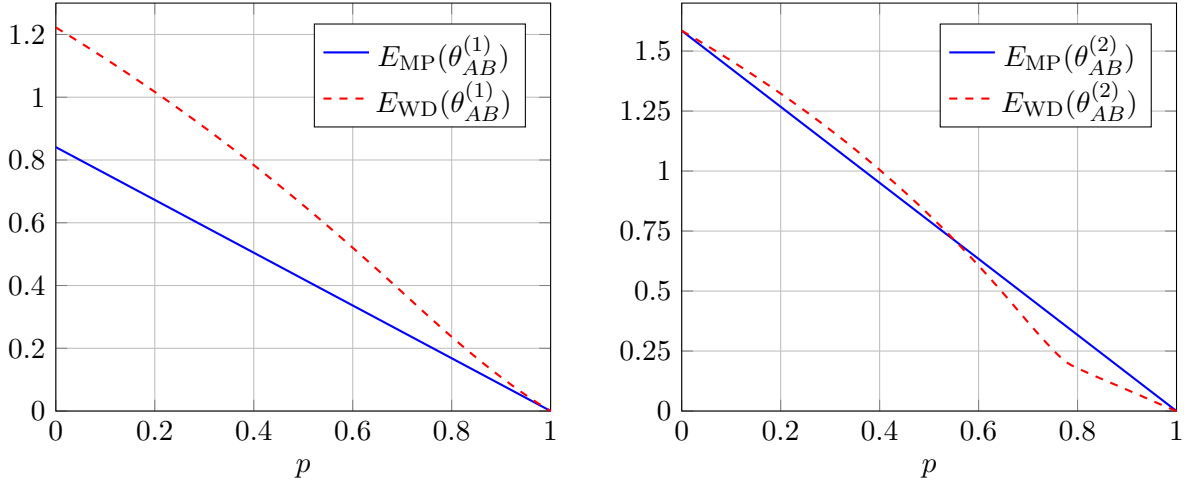
\begin{figure}[t]
	\centering
	\begin{tikzpicture}[baseline]
	\begin{axis}[name=plot1,scale=0.95,xlabel=$p$, xmin=0,xmax=1,grid = major,ymin=0,ymax=1.3,legend style = {at = {(0.95,0.95)},anchor = north east},ytick={0,0.2,0.4,0.6,0.8,1,1.2}]
	\addplot[color=blue,thick] table[x=p,y=u] {d3-example.dat};
	\addplot[color=red,dashed,thick] table[x=p,y=w] {d3-example.dat};
	\legend{$\Emp(\theta_{AB}^{(1)})$,$\Ewd(\theta_{AB}^{(1)})$};
	\end{axis}
	\begin{axis}[name=plot2,scale=0.95,anchor = south west, at=(plot1.right of south east),xlabel=$p$,xmin=0,xmax=1,grid = major,ymin=0,ymax=1.7,legend style = {at = {(0.95,0.95)},anchor = north east},xshift = 1.5cm, ytick={0,0.25,0.5,0.75,1,1.25,1.5}]
	\addplot[color=blue,thick] table[x=p,y=u] {d3-example2.dat};
	\addplot[color=red,dashed,thick] table[x=p,y=w] {d3-example2.dat};
	\legend{$\Emp(\theta_{AB}^{(2)})$,$\Ewd(\theta_{AB}^{(2)})$};
	\end{axis}
	\end{tikzpicture}
	\caption{Plot of $\Emp(\theta_{AB}^{(k)})$ from \Cref{thm:two-way-dist-upper-bound} (blue, solid) and $\Ewd(\theta_{AB}^{(k)})$ \cite{WD16} given in \eqref{eq:wang-duan} (red, dashed) as a function of $p$ for $k=1$ (left) and $k=2$ (right), where $\theta_{AB}^{(k)}$ for $k=1,2$ are defined in \eqref{eq:theta}. Both quantities are upper bounds on the two-way distillable entanglement $\Dtwo(\theta_{AB})$.}
	\label{fig:d3-example}
\end{figure}

\section{Exploiting symmetries}\label{sec:symmetries}

In this section, we derive special forms of the upper bound $\Eda(\cdot)$ on the one-way distillable entanglement (\Cref{thm:one-way-dist-upper-bound}), and of the upper bound $\Emp(\cdot)$ on the two-way distillable entanglement (\Cref{thm:two-way-dist-upper-bound}), respectively, when evaluated on states with symmetries.
In particular, we focus on the classes of isotropic and Werner states \cite{Wer89}.
To this end, we first demonstrate how both $\Eda(\cdot)$ and $\Emp(\cdot)$ can be understood as convex roof extensions.
We then exploit a theorem by \textcite{VW01} that simplifies the calculation of such convex roof extensions under a given symmetry, and apply these results to isotropic and Werner states.

\subsection{Bounds on distillable entanglement as convex roof extensions}\label{sec:convex-roof-with-symmetries}
We first review convex roof extensions of a function.
Let $K$ be a compact convex set, $M\subset K$ an arbitrary subset, and $\vphi\colon M \rightarrow \barR \coloneqq \mathbb{R}\cup\lbrace\infty\rbrace$ a function.
The \emph{convex roof} $\hat{\vphi}$ of $\vphi$ on $K$ is defined as
\begin{align}
\hat{\vphi}\colon K &\longrightarrow \barR\\
x &\longmapsto \inf \left\lbrace \sumi_i \lambda_i \vphi(m_i)\colon m_i\in M, \lambda_i \geq 0 \text{ for all $i$}, \sumi_i \lambda_i=1, \sumi_i\lambda_i m_i = x \right\rbrace\!.
\end{align}

Denoting the sets of degradable, antidegradable, maximally correlated, and positive partial transpose states by $\dgb, \adg, \mcs$, and $\ppt$, respectively, we have
\begin{align}
\Eda(\rho_{AB}) & = \inf \left\lbrace \sum_{i\colon \omega_i\in\dgb} p_i I(A\rangle B)_{\omega_i}\colon \rho_{AB} = \sumi_i p_i \omega_i \text{ with } \omega_i\in\dgb\cup\adg \text{ for all $i$.} \right\rbrace\\
\Emp(\rho_{AB}) & = \inf \left\lbrace \sum_{i\colon \omega_i\in\mcs} p_i I(A\rangle B)_{\omega_i}\colon \rho_{AB} = \sumi_i p_i \omega_i \text{ with } \omega_i\in\mcs\cup\ppt \text{ for all $i$.} \right\rbrace\!.
\end{align}
Choosing $K$ as the set of bipartite quantum states (which is convex and compact) and $\vphi(\rho_{AB})=\max\lbrace I(A\rangle B)_\rho, 0\rbrace$, it follows that both quantities can be regarded as the convex roof extension of $\vphi$ for different choices of the subset $M$:
\begin{align}
\hat{\vphi}(\rho_{AB}) = \begin{cases}
\Eda(\rho_{AB}) & \text{for } M = \dgb\cup\adg\\
\Emp(\rho_{AB}) & \text{for } M = \mcs\cup\ppt.
\end{cases}
\end{align}

We now consider states that are invariant under a given symmetry group.
First, we introduce some notation.
Let $G$ be a compact group, and let $K$ be a set with a $G$-action\footnote{For a group $G$ and a set $K$, a $G$-action on $K$ is a map $G\times K \to K$, $(g,k) \mapsto g\cdot k$ satisfying $(gh)\cdot k = g\cdot(h\cdot k)$ for all $g,h\in G$ and $k\in K$, and $e\cdot k = k$ for all $k\in K$ and the identity element $e$ of $G$.}
\begin{align}
G\times K\ni (g,k)\longmapsto g\cdot k \in K
\end{align} 
that preserves convex combinations, i.e., $g\cdot (\lambda x + (1-\lambda)y) = \lambda g\cdot x + (1-\lambda) g\cdot y$ for $x,y\in K$ and $\lambda\in[0,1]$.
Denoting the Haar measure on $G$ by $d g$, we define the $G$-twirl
\begin{align}
\cT_G(x) \coloneqq \int_G dg\, g\cdot x,
\end{align}
and we denote by $\cT_G(K)\coloneqq \lbrace k\in K\colon \cT_G(k) = k\rbrace$ the set of all $G$-invariant elements in $K$.
For any function $\vphi\colon M\to \barR$, we define the following function on $G$-invariant elements:
\begin{align}
\begin{aligned}
\vphi_G \colon \cT_G &\longrightarrow \barR\\
x &\longmapsto \inf\left\lbrace \vphi(y)\colon y\in M, \cT_G(y) = x\right\rbrace\!.
\end{aligned}\label{eq:h_G}
\end{align}
The main result we employ is the following theorem by \textcite{VW01} (see also \cite{TV00}).
\begin{theorem}[\cite{VW01}]\label{thm:symmetries}
	Let $G$ be a compact group with an action on a compact convex set $K$ that preserves convex combinations, and let $\vphi\colon M\to\barR$ be a function defined on an arbitrary subset $M$ of $K$.
	Furthermore, assume that $G\cdot M\subset M$ and $\vphi(g\cdot x) = \vphi(x)$ for all $g\in G$ and $x\in M$.
	Then for all $x\in\cT_G(K)$, 
	\begin{align}
	\hat{\vphi}(x) = \hat{\vphi}_G(x).
	\end{align}
	In particular, if $\vphi_G$ is itself convex on $\cT_G(K)$, then $\hat{\vphi}(x) = \vphi_G(x)$.
\end{theorem}

\subsection{Isotropic states and depolarizing channels}\label{sec:symmetries-isotropic}
We choose $G=\cU(d)$, the unitary group on $\mathbb{C}^d$, and consider the following action on $K=\cD(\cH_A\ox\cH_B)$, where $|A|=|B|=d$:
\begin{align}
\cU(d) \times \cD(\cH_A\ox\cH_B) \ni (U,\rho_{AB})\longmapsto U\cdot \rho_{AB} \coloneqq (U\ox \bar{U})\rho_{AB} (U\ox \bar{U})^\dagger.
\end{align}
This action is linear and thus preserves convex combinations.
The set of $G$-invariant states, $\cT_G(K)$, is the one-parameter family  $\lbrace I_d(f)\colon f\in[0,1]\rbrace$ of isotropic states:
\begin{align}
I_d(f) \coloneqq f \Phi_+ + \frac{1-f}{d^2-1} (\one_{d^2}-\Phi_+).
\end{align}
We have $\cT_G(\rho_{AB}) = I_d(f)$ with $f = \langle\Phi^+ |\rho_{AB}|\Phi^+\rangle$ for all $\rho_{AB}$.
Hence, setting 
\begin{align}
\vphi(\rho_{AB}) &= \max\lbrace I(A\rangle B)_\rho, 0\rbrace\\
M &=\dgb\cup\adg,
\end{align} 
we can compute $\Eda(I_d(f))$ by first computing $\vphi_G(I_d(f))$ and then taking the convex hull of this function, which coincides with the convex roof $\hat{\vphi}_G$.

The $d$-dimensional isotropic state $I_d(f)$ is the Choi state of the qudit depolarizing channel
\begin{align}
\cD_{p}(\rho) = (1-p) \rho + \frac{p}{d^2-1} \sum_{\substack{0\leq i,j\leq d-1 \\ (i,j)\neq (0,0)}} X^i Z^j \rho (X^i Z^j)^\dagger,
\end{align} 
where $p=1-f$ and $X,Z$ are the generalized Pauli operators defined in \eqref{eq:gen-pauli-operators}.
Since the depolarizing channel is teleportation-simulable,\footnote{
	Here, we call a channel \emph{teleportation-simulable}, if the action of the channel can be simulated by a teleportation protocol between Alice and Bob using the Choi state of the channel as an entanglement resource. 
	More precisely, a channel $\cN\colon A\to B$ with Choi state $\tau_{A'B}$ is called teleportation-simulable, if for any given input state $\rho_A$ the channel output $\cN(\rho_A)$ can be obtained by Alice and Bob performing a teleportation protocol on the joint state $\rho_A\ox \tau_{A'B}$, where $A$ and $A'$ are with Alice, and $B$ is with Bob.
	Note that Alice and Bob can establish the Choi state $\tau_{A'B}$ between them by Alice sending one half of a maximally entangled state through the channel.

	Teleportation-simulable channels in the above sense were called \emph{Choi-stretchable} in \cite{PLOB15}.
	There, the authors also consider more general simulation protocols of quantum channels, consisting of trace-preserving LOCC operations acting on the input state to the channel and a resource state shared between the two parties.}
its quantum capacity is equal to the one-way distillable entanglement of its Choi state \cite{BDSW96},
\begin{align}
Q(\cD_{1-f}) = \Done(I_d(f)).
\label{eq:teleportation-simulation}
\end{align}
Note that the quantum capacity does not increase under the assistance by forward classical communication \cite{BDSW96,BKN00}.
Hence, evaluating our upper bound $\Eda(I_d(f))$ directly yields an upper bound on $Q(\cD_{1-f})$.

\textcite{JV13} proved that $I_d(f)$ is symmetrically extendible (and hence antidegradable by \Cref{lem:antideg-symm-ext}) for $f\leq \frac{1+d}{2d}$, and hence, $\Eda(I_d(f)) = 0$ for $f\leq \frac{1+d}{2d}$.
Note that this was first proved for $d=2$ by \textcite{BDEFMS98}.
Moreover, \Cref{lem:antideg-overlap} in \Cref{sec:antidegradable-states} shows that the maximal overlap of any antidegradable state with the maximally entangled state is at most $\frac{1+d}{2d}$.
Consequently, for $f\geq \frac{1+d}{2d}$ we can restrict $M$ in the definition of $\vphi_G$ in \eqref{eq:h_G} to $\dgb$, the set of degradable states.

Numerics for low dimensions ($d=2,3$) suggest that $\vphi_G$ is convex on $\cT_G(K) = \lbrace I_d(f)\colon f\in[0,1]\rbrace$ as a function of $f$, indicating that taking the convex roof in the following theorem is not necessary (we interpret $I_d(f)$ as the Choi state of the qudit depolarizing channel $\cD_p$ with $p=1-f$):
\begin{theorem}\label{thm:one-way-depol}
	For $d\in\mathbb{N}, d\geq 2$ and $0\leq p\leq \frac{d-1}{2d}$, we have the following upper bound on the quantum capacity of the qudit depolarizing channel $\cD_p$: 
	\begin{align}
	Q(\cD_p) \leq \hat{\vphi}_G(p),
	\end{align} 
	where the function $\vphi_G(p)$ is defined as
	\begin{align} 
	\vphi_G(p)\coloneqq \inf\left\lbrace I(A\rangle B)_\rho\colon \rho_{AB}\in\dgb, \langle\Phi^+|\rho_{AB}|\Phi^+\rangle = 1-p\right\rbrace\!.
	\end{align}
\end{theorem}
Note that computing $\vphi_G(p)$ is a non-convex optimization problem, since the set $\dgb$ is not convex. 
However, for low dimensions we can still solve this problem numerically.
For $d=2$ we use the normal form of degradable quantum channels derived by \textcite{WP07} to efficiently carry out the optimization in \Cref{thm:one-way-depol}.
We apply the result from \cite{WP07} to states by interpreting a bipartite state as the Choi state of a CP, but not necessarily TP map.
Hence, we consider states of the form
\begin{align}
\begin{aligned}
\rho_{AB} &= \frac{2}{r_1^2+r_2^2} \left( (\one_A\ox K_1)\Phi^+ (\one_A\ox K_1)^\dagger + (\one_A\ox K_2)\Phi^+ (\one_A\ox K_2)^\dagger \right) \\
\text{where} \quad K_1 &= \begin{pmatrix} r_1 \cos \alpha & 0\\ 0 & r_2 \cos\beta \end{pmatrix}\!, \quad 
K_2 = \begin{pmatrix} 0 & r_2 \sin\beta \\ r_1 \sin\alpha & 0 \end{pmatrix}
\end{aligned}
\label{eq:qubit-qubit-deg-antideg}
\end{align}
for $\alpha,\beta\in\mathbb{R}$ and $0\leq r_1,r_2\leq 1$.
By an extension of the result (about channels) in \cite{WP07} to states, any quantum state $\rho_{AB}$ of the form in \eqref{eq:qubit-qubit-deg-antideg} is degradable or antidegradable.
Moreover, for these states, the condition $\langle\Phi^+|\rho_{AB}|\Phi^+\rangle = 1-p$ in \Cref{thm:one-way-depol} is equivalent to 
\begin{align}
\frac{ (r_1\cos\alpha + r_2 \cos\beta)^2}{2(r_1^2 + r_2^2)} = 1-p.\label{eq:qubit-qubit-overlap-condition}
\end{align}
The minimization of $I(A\rangle B)_{\rho}$ over states $\rho_{AB}$ of the form \eqref{eq:qubit-qubit-deg-antideg} satisfying condition \eqref{eq:qubit-qubit-overlap-condition} can be carried out using MatLab's \texttt{fmincon} function.
The resulting bound is plotted in \Cref{fig:depolarizing-bound}, together with the previously known upper bound on $Q(\cD_p)$ derived in \cite{SSWR15,SS08}.
We note that the upper bound derived in \cite{SSWR15}, which is based on approximate degradability of channels, is identical to the one obtained from \Cref{thm:approx-bound} based on approximate degradability of states.

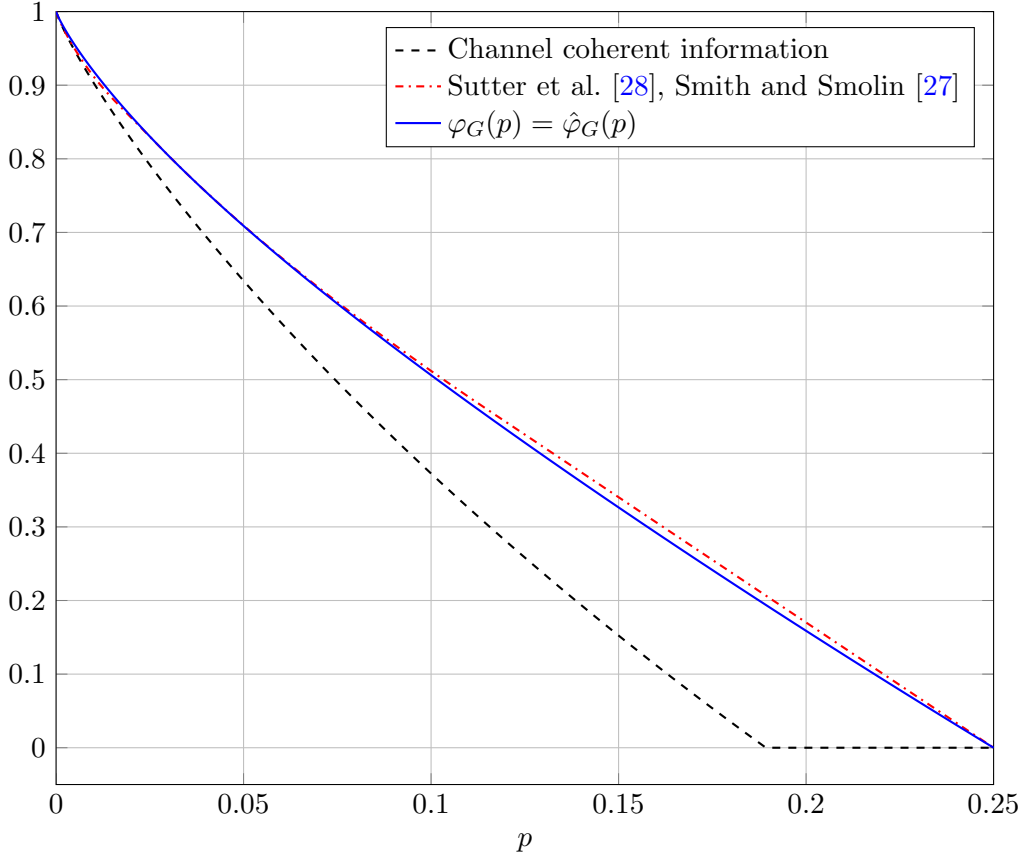
\begin{figure}[ht]
\centering
\begin{tikzpicture}
\begin{axis}[
xlabel=$p$, 
xmin=0,
xmax=0.25,
ymin=-0.05,
ymax=1, 
xtick = {0,0.05,0.1,0.15,0.2,0.25}, 
tick label style={/pgf/number format/fixed}, 
scale=1.8,
grid = major,
legend cell align = left,
]
\addplot[color=black, dashed, thick] table[x=p,y=d] {coh.dat};
\addplot[color=red, thick, dashdotted] table[x=p,y=d] {conv-hull.dat};
\addplot[color=blue,thick] table[x=p,y=d] {dist-optimized.dat};
\legend{Channel coherent information, {\textcite{SSWR15}, \textcite{SS08}}, $\vphi_G(p)=\hat{\vphi}_G(p)$};
\end{axis}
\end{tikzpicture}
\caption{Upper and lower bounds on the quantum capacity $Q(\cD_p)$ of the qubit depolarizing channel ($d=2$) for the interval $p\in[0,0.25]$. The hashing bound \eqref{eq:hashing-bound} yields the channel coherent information (black, dashed) as a lower bound on $Q(\cD_p)$. Our upper bound $\vphi_G(p)$ (blue, solid) obtained via \Cref{thm:one-way-depol} (which is convex itself and thus equal to $\hat{\vphi}_G(p)$) is compared to the upper bound obtained in \cite{SSWR15,SS08} (red, dash-dotted).
Note that the latter is identical to the upper bound on $\Done(\cJ(\cD_p))$ obtained from \Cref{thm:approx-bound}.}
\label{fig:depolarizing-bound}
\end{figure}

For $d=3$, we use another idea from \cite{WP07} to numerically optimize over degradable states.
Given a bipartite state $\rho_{AB}$ together with its complementary state $\rho_{AE}$, the degradability condition reads
\begin{align}
\rho_{AE} = (\id_A\ox \cD)(\rho_{AB}).\label{eq:degradability-condition}
\end{align}
We regard $\rho_{AB}$ and $\rho_{AE}$ as (unnormalized) Choi states of trace-non-preserving CP maps, and assign to $\rho_{AB}$, $\rho_{AE}$, and the Choi state of the degrading map $\cD\colon B\to E$ their respective transfer matrices $T(\cdot)$.\footnote{For a CP map $\cN$ with Choi state $\tau_\cN$, the \emph{transfer matrix} $T(\cN)$ is defined to be the matrix with elements $\langle ij|T(\cN)|kl\rangle = \langle lj|\tau_\cN|ki\rangle$.
	The map $T$ is a linear involution, i.e., $T^2 = \id$.}
The condition \eqref{eq:degradability-condition} then becomes
\begin{align}
T(\cD) = T(\rho_{AE}) T(\rho_{AB})^{-1} , \label{eq:degradability-transfer}
\end{align}
where we used the fact that composition and inversion of channels translate to matrix multiplication and inversion for their transfer matrices, respectively (see e.g.~\cite{Wol12}).
It follows that $\rho_{AB}$ is degradable if and only if the linear map $\cD$ defined through \eqref{eq:degradability-transfer} is CP.

To numerically carry out the optimization in \Cref{thm:one-way-depol}, we use MatLab's \texttt{fmincon} to optimize over states $\rho_{AB}$ satisfying $\langle\Phi^+|\rho_{AB}|\Phi^+\rangle = 1-p$, and for which $T(\rho_{AE}) T(\rho_{AB})^{-1}$ defines the transfer matrix of a completely positive, trace-preserving degrading map $\cD\colon B\to E$.
The resulting upper bound on $Q(\cD_p)$ is depicted in \Cref{fig:depolarizing-bound-d3}.

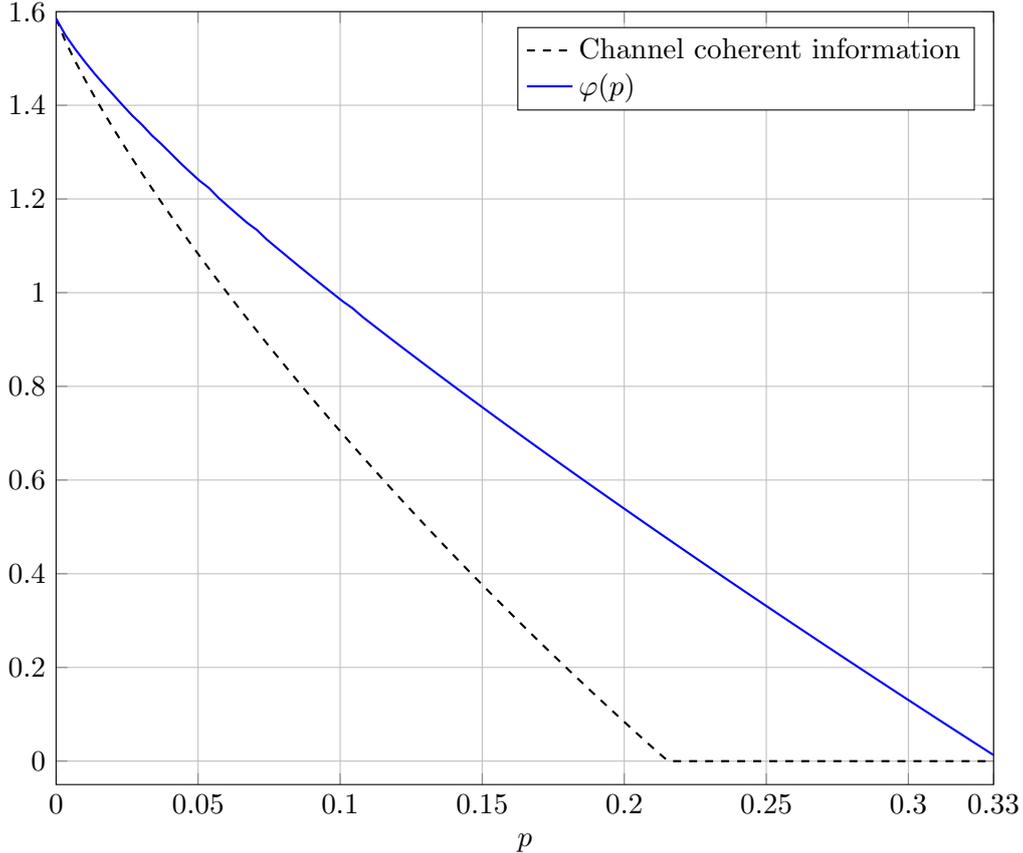
\begin{figure}[ht]
	\centering
	\begin{tikzpicture}
	\begin{axis}[
	xlabel=$p$, 
	xmin=0,
	xmax=0.33,
	ymin=-0.05,
	ymax=1.6, 
	xtick = {0,0.05,0.1,0.15,0.2,0.25,0.3,0.33}, 
	tick label style={/pgf/number format/fixed}, 
	scale=1.8,
	grid = major,
	legend cell align = left,
	]
	\addplot[color=black, dashed, thick] table[x=p,y=c] {coh3.dat};
	\addplot[color=blue,thick] table[x=p,y=u] {depol3nc.dat};

	\legend{Channel coherent information, $\varphi(p)$};
	\end{axis}
	\end{tikzpicture}
	\caption{Upper and lower bounds on the quantum capacity $Q(\cD_p)$ of the qutrit depolarizing channel ($d=3$) for the interval $p\in[0,0.33]$. The hashing bound \eqref{eq:hashing-bound} yields the channel coherent information (black, dashed) as a lower bound on $Q(\cD_p)$. The function $\varphi(p)$ obtained via \Cref{thm:one-way-depol}, which is convex up to numerical noise, is depicted in solid blue.}
	\label{fig:depolarizing-bound-d3}
\end{figure}

We now consider the two-way setting.
The best known bound on $\Dtwo(I_d(f))$ is given by the $\ppt$-relative entropy of entanglement, which for isotropic states is equal to the $\sep$-relative entropy of entanglement, and admits a particularly simple formula \cite{Rai99}:
\begin{align}
E_R^\ppt(I_d(f)) = E_R^\sep(I_d(f)) = \log d - (1-f) \log(d-1) - h(f).\label{eq:ppt-ree-isotropic}
\end{align}
In the following, we use the results from \Cref{sec:convex-roof-with-symmetries} to arrive at a different expression for this bound.

\textcite{Rai99} proved that $\langle \Phi^+|\rho_{AB}|\Phi^+\rangle \leq \frac{1}{d}$ holds for any PPT state $\rho_{AB}$.
Hence, for $f\geq \frac{1}{d}$ we can restrict to the set of maximally correlated states, $M=\mcs$.
Setting once again $\vphi(\rho_{AB}) = \max\lbrace I(A\rangle B)_\rho, 0\rbrace$, we first show that the function $\vphi_G$ defined in \eqref{eq:h_G} achieves the $\ppt$-relative entropy of entanglement \eqref{eq:ppt-ree-isotropic}, and is thus convex as a function of $f$:
\begin{lemma}\label{lem:iso}
	For all $d\geq 2$ and $f\in[0,1]$, we have 
	\begin{align}
	\vphi_G(f)\coloneqq \inf\left\lbrace I(A\rangle B)_\rho\colon \rho_{AB}\in\mcs, \langle\Phi^+|\rho_{AB}|\Phi^+\rangle = f\right\rbrace = E_R^\ppt(I_d(f)).
	\end{align}
	In particular, $\vphi_{G}(f)$ is convex in $f$.
\end{lemma}

\begin{proof}
	Consider the state
	\begin{align}
	\rho_{AB} = f \Phi_{0,0} + \frac{1-f}{d-1} \sum_{i=1}^{d-1}\Phi_{0,i},\label{eq:choice-of-rho}
	\end{align}
	where $\Phi_{0,0} = \Phi^+$.
	Since the generalized Bell states are orthogonal to each other, the state $\rho_{AB}$ satisfies $\langle\Phi^+|\rho_{AB}|\Phi^+\rangle = f$, and a simple calculation shows that
	\begin{align}
	I(A\rangle B)_\rho = \log d - (1-f)\log(d-1) - h(f).
	\end{align}
	It remains to be shown that $\rho_{AB}\in \mcs$.
	By \Cref{cor:algebraic-relation}, a mixture of the Bell states $\lbrace \Phi_{n_\alpha,m_\alpha}\rbrace_{\alpha=1,\dots,l}$ with $l\leq d$ is MC if and only if \eqref{eq:congruence-relations} holds for all $\alpha,\beta,\gamma\in [l]$.
	Since $(n_1,m_1)=(0,0)$ in our situation, \eqref{eq:congruence-relations} reduces to $n_\alpha m_\beta = m_\alpha n_\beta \mod d$ for all $\alpha,\beta\in [l]$, and these conditions are easily verified for the state $\rho_{AB}$ defined in \eqref{eq:choice-of-rho} with $\lbrace (n_\alpha,m_\alpha)\rbrace_{\alpha=1,\dots d} = \lbrace (0,0), \dots, (0,d-1)\rbrace$.
	
	Convexity of $\vphi_G(f)$ now follows from convexity of $E_R^\ppt(I_d(f))$ in $f$, which can be seen as follows. 
	First, note that for any $\lambda\in [0,1]$, we have
	\begin{align}
	\lambda I_d(f_1) + (1-\lambda) I_d(f_2) = I_d(\lambda f_1 + (1-\lambda) f_2).
	\end{align}
	Consider then the following:
	\begin{align}
	E_R^\ppt(I_d(\lambda f_1 + (1-\lambda) f_2)) &= \inf_{\sigma\in\ppt} D(\lambda I_d(f_1) + (1-\lambda) I_d(f_2) \| \sigma)\\
	&\leq D(\lambda I_d(f_1) + (1-\lambda) I_d(f_2) \| \lambda\sigma_1 + (1-\lambda)\sigma_2)\\
	&\leq \lambda D(I_d(f_1)\|\sigma_1) + (1-\lambda) D(I_d(f_2)\|\sigma_2)\\
	&= \lambda E_R^\ppt(I_d(f_1)) + (1-\lambda) E_R^\ppt(I_d(f_2)),
	\end{align}
	where in the first inequality we considered $\ppt$-states $\sigma_i$ optimizing the $\ppt$-relative entropy of $I_d(f_i)$ for $i=1,2$, respectively, and in the second inequality we used joint convexity of the quantum relative entropy.
\end{proof}

Hence, we arrive at the following result:

\begin{theorem}\label{thm:two-way-iso}
	For $d\in\mathbb{N}, d\geq 2$ and $f\geq \frac{1}{d}$, 
	\begin{align}
	\Dtwo(I_d(f)) \leq \inf\left\lbrace I(A\rangle B)_\rho\colon \rho_{AB}\in\mcs, \langle\Phi^+|\rho_{AB}|\Phi^+\rangle = f\right\rbrace\!.
	\end{align}
\end{theorem}

\subsection{Werner states}
We again set $G=\cU(d)$ and consider the following action of $G$ on $K = \cD(\cH_A\ox\cH_B)$:
\begin{align}
\cU(d)\times \cD(\cH_A\ox\cH_B) \ni (U,\rho_{AB})\longmapsto U\cdot \rho_{AB} \coloneqq (U\ox U)\rho_{AB} (U\ox U)^\dagger.
\end{align}
The set of $G$-invariant states is the one-parameter family of Werner states $\lbrace W_d(p)\colon p\in [0,1]\rbrace$ \cite{Wer89} with
\begin{align}
W_d(p) \coloneqq \frac{1-p}{d^2+d} (\one_{d^2} + \F_d) + \frac{p}{d^2-d} (\one_{d^2} - \F_d), \label{eq:werner}
\end{align}
where $\F_d \coloneqq \sum_{i,j=0}^{d-1} |i\rangle \langle j| \ox |j \rangle \langle i|$ is the swap operator on $\mathbb{C}^d\ox\mathbb{C}^d$.
We have for all $\rho_{AB}$ that $\cT_G(\rho_{AB}) = W_d(p)$ with $p = \frac{1}{2}(1-\tr(\F_d \rho_{AB}))$.

For Werner states, the best known upper bound on the two-way distillable entanglement $\Dtwo(W_d(p))$ is Rains' bound on $D_\Gamma(W_d(p))$ \cite{Rai01}:
\begin{align}
D_\Gamma(W_d(p)) \leq \begin{dcases}
0 & 0 \leq p \leq \frac{1}{2}\\
1-h(p) & \frac{1}{2} \leq p \leq \frac{1}{2} + \frac{1}{d}\\
\log\left(\frac{d-2}{d}\right) + p \log\left(\frac{d+2}{d-2}\right) & \frac{1}{2} + \frac{1}{d} \leq p \leq 1.
\end{dcases} \label{eq:rains-bound}
\end{align}
Note that for $p\geq 1/2$ we have 
\begin{align}
E_R^\ppt(W_d(p)) = \min_{\sigma\in\ppt} D(W_d(p)\|\sigma) = D(W_d(p)\|W_d(1/2)) = 1-h(p),\label{eq:werner-ppt-rel-ent}
\end{align}
since the Werner states are $\ppt$ if and only if $p\in[0,1/2]$, and for $p\geq 1/2$ the $\ppt$-Werner state closest (in relative entropy distance) to $W_d(p)$ is $W_d(1/2)$.
The expression for $E_R^\ppt(W_d(p))$ in \eqref{eq:werner-ppt-rel-ent} was proved by \textcite{VPRK97} for the case $d=2$.
Moreover, \textcite{VW01} derived an expression for the $\sep$-relative entropy of entanglement of Werner states in arbitrary dimensions that is identical to \eqref{eq:werner-ppt-rel-ent} and implies it by the argument above.
We see that Rains' bound on $D_\Gamma(W_d(p))$ is equal to $E_R^\ppt(W_d(p))$ for $\frac{1}{2} \leq p \leq \frac{1}{2} + \frac{1}{d}$, and strictly tighter for $p> \frac{1}{2} + \frac{1}{d}$.

We have for any PPT state $\sigma_{AB}$ that $\tr(\F_d \sigma_{AB}) \geq 0$ and hence $\frac{1}{2}(1-\tr(\F_d \sigma_{AB})) \leq \frac{1}{2}$, since $\F_d = d (\Phi^+)^{\Gamma_B}$ and $\sigma_{AB}^{\Gamma_B}\geq 0$ by assumption.
Therefore, we set $M=\mcs$ for $p\in(\frac{1}{2},1)$, and $\vphi(\rho_{AB}) = \max\lbrace I(A\rangle B)_\rho, 0\rbrace$ as before.
As for the isotropic states, we first show that the function $\vphi_{G}$ in \eqref{eq:h_G} achieves the $\ppt$-relative entropy of entanglement, and is thus convex:
\begin{lemma}\label{lem:werner}
	For all $d\geq 2$ and $\frac{1}{2}\leq p \leq 1$, we have
	\begin{align}
	\vphi_{G}(p)\coloneqq \inf\left\lbrace I(A\rangle B)_\rho\colon \rho_{AB}\in\mcs, \tr(\F_d \rho_{AB}) = 1-2p \right\rbrace = E_R^\ppt(W_d(p)).
	\end{align}
	In particular, $\vphi_{G}(p)$ is convex in $p$.
\end{lemma}

\begin{proof}
	We consider the state
	\begin{align}
	\rho_{AB} &= (1-p) \Psi_+ + p \Psi_-,
	\end{align}
	where $|\Psi_\pm\rangle = \frac{1}{\sqrt{2}}(|01\rangle \pm |10\rangle)$ satisfying $\F_d|\Psi_\pm\rangle = \pm |\Psi_\pm\rangle$.
	Hence, $\tr(\F_d \rho_{AB}) = 1-2p$, and furthermore $I(A\rangle B)_\rho = 1-h(p)$.
	Moreover, $\rho_{AB}$ is MC as a mixture of two Bell states \cite{HH04}, which concludes the proof of the equality.
	The convexity of $\vphi_{G}(p)$ in $p$ follows in the same way as in \Cref{lem:iso}.
\end{proof}

We thus arrive at the following result:
\begin{theorem}\label{thm:two-way-werner}
	For $d \geq 2$ and $p\geq \frac{1}{2}$, we have the following upper bound on the two-way distillable entanglement of Werner states:
	\begin{align}
	\Dtwo(W_d(p)) \leq \inf\left\lbrace I(A\rangle B)_\rho\colon \rho_{AB}\in\mcs, \tr(\F_d \rho_{AB}) = 1-2p \right\rbrace.
	\end{align}
\end{theorem}

\section{Concluding remarks}\label{sec:conclusion}
In this paper we derived upper bounds on the one-way and two-way distillable entanglement of bipartite quantum states.
In both settings we identified `useful' classes of states for which the regularized formulae for $\Done(\cdot)$ resp.~$\Dtwo(\cdot)$ reduce to the coherent information $I(A\rangle B)_\rho$, and thus a single-letter formula. 
These useful states are given by degradable and maximally correlated states, respectively.
Moreover, we identified `useless' states for which the distillable entanglement is always zero.
These are the antidegradable and PPT states, respectively.
Our upper bounds on the distillable entanglement follow from the fact that in both the one-way and two-way LOCC setting it is convex on convex combinations of useful and useless states.
The bounds are similar in spirit to the additive extensions bounds in \cite{SS08}, and always at least as good an upper bound as the entanglement of formation.
We also extended our method to obtain an upper bound on the quantum capacity based on decompositions of a quantum channel into degradable and antidegradable completely positive maps, recovering a result by \textcite{Yang}.

By interpreting our upper bounds as convex roof extensions, we were able to formulate the upper bounds on the distillable entanglement of isotropic and Werner states as a non-convex optimization problem.
For the one-way distillable entanglement of the Choi state of the qubit depolarizing channel, this optimization led to an upper bound on its quantum capacity that is strictly tighter than the best previously known upper bound for large values of the depolarizing parameter.
For the two-way distillable entanglement of both isotropic and Werner states, the non-convex optimization achieves the respective $\ppt$-relative entropy of entanglement, and thus provides new expressions for the latter.

Comparing the one-way and two-way LOCC settings with respect to how our upper bound performs in comparison to previously known upper bounds on the distillable entanglement, we notice the following discrepancy: While we get a strictly tighter bound in the one-way setting in certain cases, we can only achieve the $\ppt$-relative entropy of entanglement in the two-way setting.
This leads us to the following question: Can we develop an extended theory of one-way entanglement distillation in the same spirit as Rains' work on the $\ppt$-distillable entanglement in \cite{Rai99,Rai01}?
One possible approach to develop such a theory could be to augment the class of allowed operations from one-way LOCC to antidegradability-preserving operations.

\paragraph{Acknowledgments}
We would like to thank Christian Majenz and Andreas Winter for helpful discussions during the workshop ``Beyond IID in Information Theory'', July 18-22, 2016 in Barcelona, as well as Dong Yang, for pointing out to us the possibility of extending our method to quantum channels.
We also thank Johannes Bausch, Mark Girard and Will Matthews for helpful feedback.
This material is based upon work supported by the National Science Foundation under Grant Number 1125844.

\appendix

\section{Maximal overlap of antidegradable states with an MES}\label{sec:antidegradable-states}
In order to prove \Cref{lem:antideg-overlap}, we need the following result characterizing antidegradable states:

\begin{lemma}[\cite{Myh10}]\label{lem:antideg-symm-ext}
	A state $\rho_{AB}$ is antidegradable if and only if it has a symmetric extension.
\end{lemma}
Using \Cref{lem:antideg-symm-ext}, we can formulate the overlap of antidegradable states with the maximally entangled state as an SDP, for which strong duality holds:
\begin{lemma}\label{lem:antideg-overlap}
	The maximal overlap of an antidegradable state $\rho_{AB}$ with the maximally entangled state $\Phi^+$ can be formulated as the following SDP:
	\begin{align}
	\text{\normalfont maximize: } &\frac{1}{2} \tr(\rho_{ABB'}(\Phi_{AB}\ox \one_{B'} + \Phi_{AB'}\ox \one_B))\\
	\text{\normalfont subject to: } &\tr (\rho_{ABB'})=1\\
	&\rho_{ABB'}\geq 0.
	\end{align}
\end{lemma}

\begin{proof}
	By \Cref{lem:antideg-symm-ext}, the state $\rho_{AB}$ is antidegradable if and only if there is an extension $\rho_{ABB'}$ with $B'\cong B$ satisfying $\tr_{B'}\rho_{ABB'} = \rho_{AB}$ and $\F \rho_{ABB'}\F = \rho_{ABB'}$, where $\F$ is the swap operator on the $BB'$ system.
	Hence, we can write any antidegradable state $\rho_{AB}$ as
	\begin{align}
	\rho_{AB} = \tr_{B'}\trho_{ABB'},
	\end{align}
	where $\trho_{ABB'} = \frac{1}{2}(\rho_{ABB'} + \F \rho_{ABB'}\F)$ for some arbitrary (not necessarily symmetric) state $\rho_{ABB'}$.
	Substituting this in $\tr(\rho_{AB}\Phi_{AB}) = \tr(\trho_{ABB'}(\Phi_{AB}\ox\one_{B'}))$ then gives the SDP in the lemma.
\end{proof}

The solution of the dual problem in \Cref{lem:antideg-overlap} is equal to the largest eigenvalue of the operator $\frac{1}{2} (\Phi_{AB}\ox \one_{B'} + \Phi_{AB'}\ox\one_B)$, and therefore equal to $\frac{1+d}{2d}$ \cite{JV13}.
Hence, this is the maximal overlap of any antidegradable state with the maximally entangled state.

\emergencystretch=3em
\printbibliography[title={References},heading=bibintoc]
\end{document}